\documentclass[11pt]{article}
\usepackage[margin=1in]{geometry}

\usepackage{amsthm}
\usepackage{amsmath}
\allowdisplaybreaks

\usepackage{mathtools}

\usepackage{amssymb}
\usepackage{bbm}
\usepackage{xcolor}
\usepackage{comment}
\usepackage{standalone}
\usepackage{tikz}
\usepackage[T1]{fontenc}
\usepackage[inline]{enumitem}
\usepackage{xfrac}
\usepackage{array}
\usepackage[vlined,ruled,algo2e,linesnumbered]{algorithm2e}
\SetAlgoSkip{bigskip}
\SetAlgoInsideSkip{smallskip}

\definecolor{darkblue}{rgb}{0,0,0.38}
\definecolor{darkred}{rgb}{0.6,0,0}
\definecolor{darkgreen}{rgb}{0.1,0.35,0}

\usepackage{hyperref}
\hypersetup{colorlinks, linkcolor=darkblue, citecolor=darkgreen, urlcolor=darkblue}

\usepackage{aliascnt}

\usepackage[
backend=biber,
style=alphabetic,
citestyle=alphabetic,
maxalphanames=4,
maxcitenames=99,
mincitenames=98,
maxbibnames=99,
giveninits=true,
]{biblatex}

\usepackage[capitalize, nameinlink]{cleveref}
\crefname{appendix}{Appendix}{Appendices}
\Crefname{appendix}{Appendix}{Appendices}
\crefname{subappendix}{Appendix}{Appendices}
\Crefname{subappendix}{Appendix}{Appendices}

\usepackage{bm}

\makeatletter
\ifdefstring{\blx@bblversion}{3.3}
  {\appto{\blx@bblstart}{
    \def\entry#1#2#3{\blx@bbl@entry{#1}{#2}{#3}{0}}
    \def\datalist[#1]#2{\blx@bbl@dlist[#1]{#2/global}}
  }}
  {\ifdefstring{\blx@bblversion}{3.2}
    {}
    {\PackageError{main-mdp-arxiv}{Unsupported bibliography data format}
      {This manuscript supports biblatex BBL formats 3.2 and 3.3.}}}
\patchcmd{\blx@bblfile}
  {\blx@ifsigned{\jobname}{bbl}}
  {\@firstoftwo}
  {}
  {\PackageError{main-mdp-arxiv}{Cannot load the embedded bibliography}
    {The installed biblatex version is incompatible with this manuscript.}}
\makeatother

\usetikzlibrary{arrows}
\pgfarrowsdeclarecombine*{spaced stealth'}{spaced stealth'}{stealth'}{stealth'}{space}{space}
\tikzset{>=spaced stealth'}

\mathchardef\mhyphen="2D

\newcommand*{\eps}{\varepsilon}
\renewcommand*{\epsilon}{\varepsilon}
\renewcommand*{\vec}[1]{\bm{#1}}
\newcommand*{\R}{\mathbb R}
\newcommand*{\RR}{\R_{\ge 0}}
\newcommand*{\Z}{\mathbb Z}
\newcommand*{\ZZ}{\Z_{\ge 0}}

\renewcommand*{\Pr}{\mathbb P}
\newcommand*{\E}{\mathbb E}
\newcommand*{\opt}{\mathrm{OPT}}

\newcommand*{\1}{\mathbf{1}}

\newcommand*{\RRR}{\mathtt{RRR}}
\newcommand*{\LH}{H}

\setlist[enumerate,2]{label=\arabic{enumi}.\arabic{enumii}, ref=\arabic{enumi}.\arabic{enumii}}
\setlist[enumerate,3]{label=\arabic{enumi}.\arabic{enumii}.\arabic{enumiii}, ref=\arabic{enumi}.\arabic{enumii}.\arabic{enumiii}}
\setlist[enumerate,4]{label=\arabic{enumi}.\arabic{enumii}.\arabic{enumiii}.\arabic{enumiv}, ref=\arabic{enumi}.\arabic{enumii}.\arabic{enumiii}.\arabic{enumiv}}

\newtheorem{theorem}{Theorem}

\newaliascnt{definition}{theorem}

\aliascntresetthe{definition}

\newaliascnt{lemma}{theorem}
\newtheorem{lemma}[lemma]{Lemma}
\aliascntresetthe{lemma}

\newaliascnt{proposition}{theorem}

\aliascntresetthe{proposition}

\newaliascnt{corollary}{theorem}

\aliascntresetthe{corollary}

\newaliascnt{property}{theorem}

\aliascntresetthe{property}

\newaliascnt{claim}{theorem}

\aliascntresetthe{claim}

\definecolor{TolDarkPurple}{HTML}{332288}
\definecolor{TolDarkBlue}{HTML}{0D4DBD}
\definecolor{TolLightBlue}{HTML}{88CCEE}
\definecolor{TolLightGreen}{HTML}{44AA99}
\definecolor{TolDarkGreen}{HTML}{117733}
\definecolor{TolDarkBrown}{HTML}{999933}
\definecolor{TolLightBrown}{HTML}{DDCC77}
\definecolor{TolDarkRed}{HTML}{661100}
\definecolor{TolLightRed}{HTML}{CC6677}
\definecolor{TolLightPink}{HTML}{AA4466}
\definecolor{TolDarkPink}{HTML}{882255}
\definecolor{TolLightPurple}{HTML}{AA4499}

\newcommand\mperiod{\rlap{\,.}}
\newcommand\mcomma{\rlap{\,,}}

\title{Strong and Compact Policies for Submodular Markov Decision Processes via LP-Based Submodular Orienteering}

\author{Lars Rohwedder\footnote{University of Southern Denmark, Odense, Denmark} \and Rico Zenklusen\footnote{ETH Zurich, Zurich, Switzerland}}

\date{}

\begin{document}
\maketitle
\begin{abstract}
Finding policies for Markov Decision Processes (MDPs)
is a central problem in areas such as Reinforcement Learning and Operations Research.
Here, we have to repeatedly choose an action that should be performed by an agent.
Depending on the action and the current state of the agent, the agent collects a reward and
randomly transitions into a new state.
The goal is to maximize the reward in expectation over a finite time horizon of length $H$. 
We consider a recently introduced variant that generalizes the traditionally additive reward function in the model to a monotone submodular one, which allows for capturing a range of interesting applications.

Without the stochastic component, this problem is equivalent to the Submodular Orienteering problem,
where the goal is to find an $s$-$t$ walk in a directed graph maximizing a monotone submodular function under
a length constraint.
We present a novel LP-based algorithm for Submodular Orienteering using ideas from the Sherali-Adams hierarchy and Round-or-Cut.
Our guarantees are comparable to the known quasi-polynomial time logarithmic approximation for Submodular Orienteering,
but also extend to the setting of Submodular Markov Decision Processes.
In the polynomial time regime, we present an $O(n^\eps)$-approximation (and $O(H^\eps)$ for Submodular MDPs) for every
$\eps > 0$, where $n$ is the number of vertices, which was unknown even for Submodular Orienteering.
Prior to our work, the best known approximation guarantee for Submodular MDPs had an approximation ratio linear in~$H$.

Beyond these algorithmic results, our methods reveal a trade-off between the approximation guarantee and the number of previously visited vertices on which an agent conditions its decision.
 \end{abstract}

\section{Introduction}
The Submodular Orienteering problem asks to maximize a non-negative monotone submodular function over the vertices of a walk in a directed graph $G=(V,A)$ from a vertex $s$ to a vertex $t$.\footnote{A function $f\colon 2^U \to \R$ is submodular if for all $A\subseteq B \subseteq U$ and $u\in U\setminus B$, it satisfies $f(A \cup \{u\}) - f(A) \ge f(B\cup \{u\}) - f(B)$.
It is monotone if $f(A) \le f(B)$ for all $A\subseteq B \subseteq U$. Throughout this work, we always assume that the submodular functions we consider are non-negative. Moreover, they are assumed to be given implicitly through a value oracle, which, for a given input set $S$, returns the value $f(S)$ in polynomial time.}
Additionally, the length of the walk may be constrained via a length function $w: A\rightarrow \ZZ$, which
should not exceed a given budget $C$ when summed over all edges in the walk.
We denote by $f(W)$ and $w(W)$, respectively, the submodular objective value and the length of the walk $W$.
The Recursive Greedy algorithm by \textcite{chekuri2005recursive} is an elegant combinatorial
algorithm that achieves an optimal (under natural complexity assumptions) logarithmic approximation for it in quasi-polynomial time.

In this paper we first develop an LP-based randomized rounding that recovers the same guarantees (in expectation)
up to constant factors as~\textcite{chekuri2005recursive} and an extension that gives a polynomial time $O(n^\epsilon)$-approximation for any $\epsilon > 0$.
\begin{theorem}\label{thm:main-orienteering}
	There is an LP-based approximation algorithm for Submodular Orienteering that
	runs in time $n^{O(\log n)} \langle w \rangle^{O(1)}$ and outputs an $s$-$t$ walk $W$ with $\E[f(W)] \ge \Omega(\frac{\opt}{\log n})$
	and $\E[w(W)] \le C$.
	Alternatively, it can output in time $n^{O(\sfrac{1}{\epsilon})} \langle w \rangle^{O(1)}$ an $s$-$t$ walk with $\E[f(W)] \ge \Omega(\frac{\opt}{n^\epsilon})$ and $\E[w(W)] \le C$ for any $\epsilon > 0$.\footnote{We state all running times as the number of operations in the word RAM model, where a word is large enough to store any input value, including the values returned from function $f$.}
	Here, $\langle w \rangle$ is the encoding size of the function values of $w$.
\end{theorem}
The new trade-off in the polynomial time regime can also be obtained by a purely combinatorial algorithm, which we include in \Cref{sec:comb_neps}.

Our algorithm differs significantly from previous LP-based algorithms for problems involving submodular maximization,
which usually rely on the multilinear extension. (See, in particular,~\cite{vondrak2008optimal,calinescuMaximizingMonotoneSubmodular2011a,ChekuriVZ10,chekuriSubmodularFunctionMaximization2014,feldmanUnifiedContinuousGreedy2011,eneConstrainedSubmodularMaximization2016,ConstrainedSubmodularMaximization,buchbinderConstrainedSubmodularMaximization2024} and references therein.)
In our problem, the multilinear extension is extremely weak.\footnote{Consider the graph consisting of $k = \Theta(\sqrt n)$ many $s$-$t$ paths $P_1,\dotsc,P_k$ of length $k+2$ each that are disjoint except for $s$ and $t$. The vertices of $P_i\setminus \{s,t\}$ are colored with color $i$. Function $f$ counts the number of different colors. Thus, the optimal solution is $1$, but the multilinear extension on the symmetric solution that takes each path with weight $1/k$ has a value of $\Omega(k) = \Omega(\sqrt n)$. The symmetric solution is a convex combination of integral solutions, hence, using the multilinear extension with any LP relaxation that is oblivious to $f$ seems hopeless.}
Our relaxation is based on a Sherali-Adams type extended formulation. The relaxation
and the randomized rounding algorithm are
closely related to techniques known from the Directed Steiner Tree problem, the Robust Shortest Path problem
and others, 
see, e.g.,~\cite{10.1137/20M1312988, li2024polylogarithmic, bamas2026randomized},
which, however, do not involve submodular functions.
We manage to integrate submodular functions with a Round-or-Cut approach. Round-or-Cut, first introduced
in~\cite{carr2000strengthening}, is a framework where one starts with a fractional solution to a potentially
weak LP relaxation
and then either obtains a good solution by rounding it or derives new cutting planes that can be added to
the relaxation.
This approach only needs to generate cutting planes when rounding fails, which allows for using properties of the failed rounding approach to generate cutting planes.
For our case, we proceed as follows.
Along with the rounded solution, our algorithm outputs a linear upper bound
on the submodular function such that the rounded solution has an approximation guarantee with respect to the
linear function's value in the current fractional point.
If the submodular value of the rounded solution is not high enough, the linear upper bound can be used
for a cutting plane that we
can add to the relaxation and repeat.

By running the algorithm on LP solutions that satisfy additional properties, we can influence the output
distribution. 
While the LP-based approach can be of independent interest, we emphasize
a specific application in Markov Decision Processes that nicely highlights advantages of our approach and leads to a significant improvement compared to prior work.
\subsection{Markov Decision Processes}
The Markov Decision Process (MDP) is a fundamental model to describe sequential decision making under uncertainty.
It is arguably the central mathematical framework for Reinforcement Learning, which is a subfield of Machine Learning that has seen tremendous success in recent years, and it is also a classical model in Operations Research (see, e.g.,~\cite{PUTERMAN1990331, wiering2012reinforcement}). 
An MDP has a finite set of states $S$ and actions $A$.
A common variant that we focus on here assumes a finite time horizon $\{1,2,\dotsc,H\}$.
At each time $t\in\{1,2,\dotsc,H\}$, an agent is in some state $s_t\in S$.
The agent then selects an action $a_t \in A$ according to a policy, which is the object we aim to optimize.
Based on the current state $s_t$ and the action $a_t$ chosen by the agent, the environment randomly
selects the next state $s_{t+1}$. The distribution of $s_{t+1}$ is completely determined by $s_t$ and $a_t$
and it is known a priori to the agent via the probabilities
\begin{equation}\label{eq:mdp-successor}
	\Pr[s_{t+1} = s' \mid s_t = s,\ a_t = a] , \quad \text{for }t\in\{1,2,\dotsc, H-1\}, s'\in S, s\in S, a\in A \mperiod
\end{equation}
Similarly, the initial state $s_1$ is sampled with explicitly given probabilities $\Pr[s_1 = s']$ for $s'\in S$.
In the classical model, the agent collects rewards for each state and action and
wants to maximize their sum in expectation.
We consider a more general model, where the total reward is $f(\{s_1,a_1,\dotsc,s_{H}, a_H\})$ for a monotone submodular function $f\colon 2^{S\cup A} \rightarrow \RR$.
We devise a policy that selects an action based on the trajectory $s_1,a_1,s_2,a_2,\dotsc,s_t$, that is, all previous states and actions.
Note that in the classical additive case, where rewards are additive instead of submodular, there is an optimal \emph{Markovian} policy, i.e., a policy that solely depends on the current state and time. It can be found efficiently by dynamic programming using Bellman's equation.
This is not the case in our setting, as for example pointed out in~\cite{prajapat2024submodular}.
It was also observed in~\cite{prajapat2024submodular} that the deterministic
version, that is, each probability in~\eqref{eq:mdp-successor} and those of initial states are either zero or one,
corresponds exactly to the Submodular Orienteering problem.

Our model was recently introduced and studied by \textcite{wang2020planning,prajapat2024submodular}.
It is motivated by a wide range of applications from Reinforcement Learning to planning, including navigation, biodiversity monitoring, Bayesian experimental design, and coverage maximization.
We stick to the terminology used in \cite{prajapat2024submodular} and call this model \emph{Submodular Markov Decision Process (Submodular MDP)}.
Note that the dynamics of a Submodular MDP remain Markovian, but the reward function is not Markovian anymore.
Such processes are sometimes also referred to as \emph{non-Markovian Reward Decision Processes (NMRDPs)} in the literature.

The best approximation guarantee for Submodular MDP without additional assumptions is $O(H)$ from~\cite{wang2020planning}.
\textcite{pmlr-v235-de-santi24b} give other non-trivial approximations assuming the submodular function
has \emph{bounded curvature}, a parameter that measures how much the marginal value of an element can decrease and captures how far the function is from being additive.
The works above also contain empirical results and results for simplified cases.

We improve the approximation ratio from $O(H)$ to $O(\log H)$ and provide a trade-off between approximation ratio and running time.
\begin{theorem}\label{thm:main-mdp}
	Submodular Markov Decision Processes admit a randomized $O(\log \LH)$-approximate policy
	that can be computed in $n^{O(\log \LH)} \langle p \rangle^{O(1)}$ time.
	Furthermore, for any $\epsilon > 0$, there is a randomized $O(\LH^\epsilon)$-approximate policy that can be computed in $n^{O(\sfrac{1}{\epsilon})}\langle p \rangle^{O(1)}$ time.
	Each decision in the former policy depends only on $O(\log \LH)$ many vertices from the trajectory; in the latter on $O(\sfrac{1}{\epsilon})$ vertices.
	Here, $n = \mathrm{poly}(|S|, |A|, H)$ and $\langle p \rangle$ is the encoding size of the transition probabilities.
\end{theorem}
Roughly speaking, we use our LP-based algorithm for Submodular Orienteering on an LP with constraints that force certain edges to be taken with the transition probabilities from the input.
It follows from properties of our LP rounding that the distribution over walks that we obtain can be reformulated as a randomized policy together with the transition distributions.

This result falls into the area of adaptive algorithms in stochastic combinatorial optimization.
The adaptivity comes from the conditioning operation in Sherali-Adams on
previous events that the LP rounding performs.
Using linear programs to derive policies is common in stochastic optimization, see, e.g.,~\cite{dean2008approximating, guha2010approximation}, but we are not aware of previous examples using conditioning in
Sherali-Adams or another hierarchy in a similar way.

The complexity of representing history-dependent policies is a recurring issue in sequential decision making.
For MDPs with non-additive rewards, augmenting the state with relevant history can substantially enlarge the state space, and an explicit policy representation may require exponential space~\cite{baccus1996rewarding,thiebauxDecisiontheoreticPlanningNonMarkovian2006}.
For submodular MDPs, \textcite{wang2020planning} already obtain an $O(H)$-approximation with a polynomial-size policy representation, while \textcite{prajapat2024submodular} explicitly discuss the role of history dependence and the difficulty of representing general policies.
Our result achieves an $O(\log H)$-approximation with a quasi-polynomial-size policy representation, whose decisions depend on only $O(\log H)$ selected vertices of the past trajectory.
Prior to our work, it was unclear whether an $o(H)$-approximation could be achieved with a policy that has a short description.

\subsection{Other related works}
Our Sherali-Adams approach and randomized rounding are inspired by algorithmic ideas
appearing in a number of previous works
on variants of the Directed Steiner Tree problem~\cite{10.1137/20M1312988,guo2022approximating}, the Robust Shortest Path problem~\cite{li2024polylogarithmic}, the Santa Claus problem~\cite{bateni2009maxmin, chakrabarty2009allocating},
as well as others. Recently, \textcite{bamas2026randomized} gave a general framework capturing many of these examples.
The most obvious difference is that our algorithm extends these techniques by handling monotone submodular functions.
Previously this was known only for very specific submodular functions, namely, some forms of coverage functions
(capturing for example the number of different colors when elements are partitioned into color classes).
It is shown for example in~\cite{bamas2026randomized} how to obtain via an LP approach a similar 
trade-off between running time and approximation as in our result for Orienteering
when maximizing the number of colors in the solution. This would not directly lead to a policy for MDPs,
even with these restrictive functions.
This is because of other subtle differences in the randomized rounding.
To apply the algorithm to MDPs, it is important that when making a decision for one
vertex or edge, we only condition on previous vertices and not future vertices, so that the resulting policy does not depend on the future.
The randomized rounding we present in this paper has such a property, while~\cite{bamas2026randomized, li2024polylogarithmic},
the closest related algorithms, condition on the middle point of the walk and then recurse in the first
half and the second half.

Many stochastic routing problems have been considered in the literature.
Apart from the previously mentioned works on Submodular MDPs~\cite{wang2020planning,prajapat2024submodular, pmlr-v235-de-santi24b},
we are not aware of any works that
have direct implications for our setting,
so we only give a few representative examples here:
\textcite{jiang2020algorithms}
consider the problem of collecting a high reward by a tour of low cost. In their model, either the reward
or the cost is stochastic and both are additive.
This is substantially simpler than our setting where both rewards and the movement of the agent
are stochastic. 
\textcite{tan2024informative} consider a model where the agent moves in the network
and makes stochastic observations at each vertex. The agent continues until a monotone submodular function
over the observations is large enough. One can again view this as stochastic rewards without stochastic
movement. We note that~\cite{tan2024informative} also use Submodular Orienteering as a subroutine
for their policy, but in the details, their approach differs significantly from ours.
 \section{LP-Based Orienteering}
\label{sec:orienteering}
For clarity, we present our main technical result on a simple variant of Submodular Orienteering, to which we later reduce:
We assume that the graph is layered with layers $V_1\cup\cdots\cup V_{\LH}$, where
all arcs go from one layer to the next. Note that since this is a DAG, every walk is a path.
We further assume that $H = d^r$ for suitable parameters $d$ and $r$ that we will specify later. 
Our goal is to find a path $P$ that is \emph{layer-spanning}: $P$ starts at any vertex in $V_1$ and ends at any vertex in $V_{\LH}$.
For now we do not consider restrictions on the length like the ones given by $w$ and $C$ earlier.

As discussed, a key novelty in our approach is that it is LP-based.
The linear inequality description we use can naturally be interpreted as a strengthening of the following basic linear description through an extended formulation:
\begin{alignat*}{3}
	&&x_{u} + x_{v} &\le 1 &\qquad&\forall (u,v)\in (V_i\times V_{i+1}) \setminus A \text{ for some }i\in \{1,\dots,\LH-1\} \\
	&&\sum_{v\in V_i} x_{v} &= 1 &&\forall i\in\{1,\dotsc,\LH\} \\
	\raisebox{4mm}[\height]{$(Q)$}\qquad &&x_v &\le 1 &&\forall v\in V \\
																						 &&x_v &\ge 0 &&\forall v\in V \mperiod
\end{alignat*}
We denote by $Q$ the polytope defined by the basic linear description above.
The variable $x_v$ indicates whether vertex $v$ is in the path or not.
For a $\{0,1\}$-solution, the first constraint ensures that for any two consecutive vertices $u$ and $v$ in the path, there is an arc between them.
The second constraint ensures that exactly one vertex is chosen from each layer.

In a fractional solution, we think of $x_v$ as the probability that $v$ is in a random path $P$.\footnote{This is purely for intuition at this point and should not be read as a formal statement.} We would like to extend and
strengthen this relaxation by including some additional information on the correlation between variables in the linear description, specifically, on how
likely a set of vertices is to be contained in $P$ simultaneously.
Note that without a strengthening the basic relaxation described above can be extremely weak.
In particular, it can be feasible even if the graph does not contain any layer-spanning path.\footnote{For example, one can think of a $2$-layered graph without any arcs.
A feasible LP point is obtained by putting two values of $0.5$ on each layer (and zeroes everywhere else). 
}
There are simpler ways to describe the convex hull of layer-spanning paths by a compact linear programming
formulation, but the specific variables we introduce will give us more flexibility in the rounding algorithm and
to express more complicated linear constraints that we will add later.

Towards this, we write an extended formulation
parameterized by some $r\in\ZZ$. Our variables are $x_I$ for all $I\subseteq V$, $|I|\le r+1$. 
Variable $x_I$ should mimic the behavior of $\prod_{i\in I} x_i$ or, in other words, $x_I$ intuitively describes
the probability that $I\subseteq P$. Readers familiar with the Sherali-Adams hierarchy,
see~\cite{sherali1990hierarchy}, will
recognize the similarity with this construction and, indeed, our construction can be seen as a simplified version of the Sherali-Adams hierarchy.
For self-containedness and clarity, we construct the relaxation here from first principles.

For each $I \subseteq V$ with $|I|\le r$, and each constraint of the original LP, we multiply each side of each constraint by $\prod_{i\in I} x_i$.
For example, assume that $u,v,w\in V$ are all distinct. Let
$I = \{u, w\}$ and consider the constraint $x_u + x_v \le 1$. The new constraint is then
$x_u \cdot x_w \cdot x_u + x_u \cdot x_w \cdot x_v \le x_u \cdot x_w$. This constraint is clearly satisfied if
the original one was. Since binary variables $y\in\{0,1\}$ satisfy $y^2 = y$, we can remove duplicate variables
in monomials. The previous example then simplifies to $x_u \cdot x_w + x_u \cdot x_w \cdot x_v \le x_u \cdot x_w$.
This maintains feasibility if $\vec x$ is binary.
Now we replace products of variables by our new variables.
In the example, we obtain $x_{\{u,w\}} + x_{\{w, v, u\}} \le x_{\{u,w\}}$. Finally, we add the constraint $x_{\emptyset} = 1$, which corresponds to requiring the empty product to equal $1$.

We obtain the following linear description, which we denote by $Q^{(r)}_G$, or simply by $Q^{(r)}$ if $G$ is clear from context.
\begin{alignat}{4}
	&& x_{I\cup\{u\}} + x_{I\cup \{v\}} &\le x_I &\qquad&\forall (u,v)\in (V_i\times V_{i+1})\setminus A \text{ for some }i\in \{1,\dots, \LH-1\},\label{lp:edge}\\
	&& & &&\forall I\subseteq V,  |I|\le r \notag\\
&& \sum_{v\in V_i} x_{I\cup \{v\}} &= x_I &&\forall i\in\{1,\dotsc,\LH\}, \ \forall I\subseteq V, |I|\le r & \label{lp:layer}\\[-0.6em]
	\raisebox{3mm}[\height]{$(Q^{(r)})$}\qquad	&&	x_{I\cup \{v\}} &\le x_I &&\forall v\in V \ \forall I\subseteq V, |I|\le r & \notag\\
																				&&	x_\emptyset &= 1 && & \notag\\
																				&& x_I &\ge 0 &&\forall I\subseteq V, |I|\le r+1 & \notag
\end{alignat}
For simplicity of notation, we write $x_v = x_{\{v\}}$, $x_{u,v} = x_{\{u,v\}}$, etc.

Consider a layer-spanning path or, in other words, a solution $P$.
Clearly, the incidence vector $\1_{P}\in \{0,1\}^V$ of the vertices in $P$ is feasible for the basic linear description $Q=Q^{(0)}$.
For $\vec x\in [0,1]^n$, let $\vec x^{(r)} \in [0,1]^{V \choose \le r+1}$ be defined by $x^{(r)}_I = \prod_{i\in I} x_i$ for all $I\subseteq V$, $|I|\le r+1$.\footnote{Here, ${V \choose \le r+1}$ stands for $\{U\subseteq V : |U|\le r+1\}$.}
For integral solutions $\vec x\in Q$ of the basic linear description $Q$ we have $\vec x^{(r)}\in Q^{(r)}$ by the previous discussion.

A crucial feature of this extended formulation is that it allows us to condition on variables.
For some $\vec x\in Q^{(r)}$ and $v\in V$ with $x_{v} > 0$, define $\vec x^{|v} \in [0,1]^{\binom{V}{\le r}}$ by
\begin{equation*}
	x^{|v}_I \coloneqq \frac{x_{I\cup \{v\}}}{x_v},\quad I\subseteq V, |I|\le r \mperiod
\end{equation*}
Note that $x^{|v}_{v} = 1$ and $x^{|v}_{u} = x^{|u}_{v} \cdot \frac{x_{u}}{x_{v}}$ (analogous to Bayes' rule),
which matches our intuition that $\vec x^{|v}$ corresponds to conditioning the probability distribution on including $v$.
The following lemma shows that conditioning a point $\vec x \in Q^{(r)}$ on one variable leads to a point that is still feasible for the extended formulation, but with one level less, i.e., a point in $Q^{(r-1)}$.
\begin{lemma}\label{lem:condition}
	Let $r\in\Z_{\ge 1}$, $\vec x\in Q^{(r)}$, and $w\in V$ with $x_w > 0$. Then
	$\vec x^{|w}\in Q^{(r-1)}$.
\end{lemma}
\begin{proof}
	We verify this only for constraints of type~\eqref{lp:edge}, but the same argument works for all other types of constraints as well.
	Let $i\in\{1,\dotsc,\LH-1\}$, $(u,v)\in (V_i\times V_{i+1})\setminus A$, and $I\subseteq V$ with $|I|\le r-1$.
	Because $\vec x\in Q^{(r)}$ and it satisfies~\eqref{lp:edge} with $I' = I\cup \{w\}$, we have
	\begin{equation*}
		x_{I\cup\{w,u\}} + x_{I\cup\{w,v\}} \le x_{I\cup\{w\}} \mperiod
	\end{equation*}
	Multiplying both sides by $\sfrac{1}{x_{w}}$, we obtain
	\begin{equation*}
		x^{|w}_{I\cup\{u\}} + x^{|w}_{I\cup\{v\}} = \frac{x_{I\cup\{w,u\}}}{x_w} + \frac{x_{I\cup\{w,v\}}}{x_w} \le \frac{x_{I\cup\{w\}}}{x_w} = x^{|w}_{I} \mperiod
	\end{equation*}
	Therefore, $\vec x^{|w}$ satisfies~\eqref{lp:edge}.
\end{proof}
Based on this conditioning property, we now define a natural recursive randomized rounding algorithm.
The algorithm has a parameter $d\in \Z_{\ge 2}$ that controls the degree of the recursion. For lower values
of $d$ we obtain a better approximation ratio at the cost of a higher running time. We assume from here on that $\LH$
is a power of $d$, that is, $\LH = d^r$ for some $r\in \ZZ$. This value of $r$ will be the parameter
that we use in the linear program to define $Q^{(r)}$.
Given a point $\vec x\in Q^{(r)}$, the rounding algorithm returns a random layer-spanning path $P$ whose distribution has useful properties linked to $\vec x$.

This approach will allow us to later add additional constraints to $Q^{(r)}$ that are satisfied by $\1^{(r)}_{\opt}$ for the optimal solution $\opt$.
These constraints will lead to points $\vec x$ with additional properties that translate into distributional properties of $P$ that are crucial for our applications.
Moreover, as we see later, the LP-based approach also allows us to add constraints through a round-or-cut procedure for obtaining good approximation guarantees for submodular maximization.

\subsection{Recursive randomized rounding}

Let $G=(V,A)$ be a layered graph with $V = V_1 \,\cup \cdots \cup\, V_{\LH}$, $\LH = d^{r}$, and let $\vec x\in Q^{(r)}$.
We now describe our recursive randomized rounding algorithm to round a point $\vec x\in Q^{(r)}$ to a random layer-spanning path $P$.
To present the algorithm, we use the following notation.
For some set of layer indices $L\subseteq \{1,2,\dotsc,\LH\}$, we let $G[L]$ be the induced subgraph on $V[L] \coloneqq \bigcup_{i\in L} V_{i}$.
Moreover, we denote by $\vec x[L]$ the vector $\vec x$ restricted to variables of the vertex sets fully contained in $G[L]$.
We will consider only layer indices $L$ corresponding to consecutive sets of layers.
Our recursive randomized rounding algorithm is described in \Cref{alg:recursive_rounding}.

\begin{algorithm2e}
\DontPrintSemicolon
	\caption{Recursive randomized rounding (RRR)}\label{alg:recursive_rounding}
	\KwIn{Layered orienteering graph $G=(V,A)$ with $V=V_1\cup \cdots \cup V_{\LH}$, $\vec x\in Q^{(r)}$, $d^r = \LH$.}
\KwOut{Layer-spanning path in $G$.}

\uIf{$\LH=1$}{
	Sample $v\in V_1$ with $\Pr[v = w] = x_{w}$ for all $w\in V_1$. \;
	\Return $(v)$.\label{algline:hIsZero}
}
\Else{
	Let $L_i = \{(i-1) \frac{\LH}{d}+1, (i-1)\frac{\LH}{d}+2,\dotsc,i\frac{\LH}{d}\}$ for each $i\in\{1,2,\dotsc,d\}$. \;
	Recursively construct layer-spanning path $P_{1}$ in $G[L_1]$ from $\vec x[L_1]$ projected to $\R^{\binom{V}{\le r}}$.\;
	\For{$i=2,3,\dotsc,d$}{
		Let $u$ be the last vertex in $P_{i-1}$.\;
		Recursively construct layer-spanning path $P_{i}$ in $G[L_i]$ from $\vec x^{|u}[L_i]$.\;
	}
	\Return concatenation $P_1\circ P_2\circ\cdots \circ P_{d}$.
}
\end{algorithm2e}
Given a point $\vec x \in Q^{(r)}$,
denote by $\RRR(\vec x)$ the output distribution of \Cref{alg:recursive_rounding} on $\vec x$. 

Our recursive construction of a walk is somewhat reminiscent of the Recursive Greedy algorithm by \textcite{chekuri2005recursive}, but there are some important differences.
Apart from the obvious difference that we are LP-based, our algorithm, when recursively constructing parts of the path, never needs to condition on vertices that lie in the future, i.e., that come later in the path.
(We will make this precise in \Cref{sec:mdp}.)
This is crucial for later obtaining policies for MDPs, which are not allowed to depend on future states.
Conversely, the Recursive Greedy algorithm needs to condition on vertices that lie in the future.
For example, the first step is to guess the middle vertex of the path, which is a vertex that lies in the future for the first half of the path.
And this step is repeated recursively in the algorithm of \textcite{chekuri2005recursive}.

We also want to highlight that our parameter $d$ chops the path into $d$ pieces, which are then constructed recursively and in particular \emph{sequentially}.
This is different from guessing $d$ vertices of the path simultaneously, which would lead to both a much higher running time and a significant dependence on future vertices.
This technique of guessing several vertices simultaneously was discussed in the Recursive Greedy algorithm by \textcite{chekuri2005recursive} (see Section~3.3), to get slightly improved approximation ratios.

We first observe that the different steps of \Cref{alg:recursive_rounding} are well-defined and indeed lead to
a layer-spanning path being sampled.
We also prove that $\RRR(\vec x)$ is marginal-preserving on singleton variables.
\begin{lemma}\label{lem:marginal-preserving}
Consider a layered graph $G=(V,A)$ with $V = V_1 \cup \cdots \cup V_{\LH}$, $\LH = d^r$, and let $\vec x\in Q^{(r)}$.
Then recursive randomized rounding applied to $\vec{x}$
\begin{enumerate}[nosep, label=(\roman*)]
		\item is well defined;
		\item returns a layer-spanning path; and
		\item the output path $P$ satisfies $\Pr[v\in P] = x_{v}$ for all $v\in V$.\label{en:marginals}
	\end{enumerate}
\end{lemma}
\begin{proof}
We show the statement by induction over $\LH$.
	If $H = 1$, then by~\eqref{lp:layer} we have $\sum_{w\in V_1} x_{w} = x_\emptyset = 1$.
	Thus, the variables $\{x_{w}\}_{w\in V_1}$ indeed define a probability distribution for $v$.
	The output is a single vertex, which is layer-spanning in this case, and it satisfies~\ref{en:marginals} by definition.

	Now assume that $H \ge d$ and let $L_i$ be defined as in the algorithm.
	To be able to make the first recursive call, we need to verify that $\vec x[L_1]$ projected to $\R^{\binom{V_1}{\le r}}$ is in $Q^{(r-1)}_{G[L_1]}$.
	This holds trivially, because all constraints in $Q^{(r-1)}_{G[L_1]}$ are also valid for $Q^{(r)}_{G}$.
	By the induction hypothesis, a vertex $v\in V[L_1]$ is in $P_1$ and therefore
	in $P$ with probability
	$x[L_1]_v = x_v$.
	Thus,~\ref{en:marginals} holds for such vertices. 

	Next, assume that $v\in V[L_i]$ for some $i\ge 2$.
	The path returned by recursive randomized rounding in $G$ is $P = P_1\circ P_2\circ\cdots\circ P_d$. We can assume (via a nested induction) that each $w\in V_{\sfrac{(i-1)\LH}{d}}$ is the last vertex of $P_{i-1}$ with probability $\Pr[w = u] = x_{w}$ and, in particular, $x_u > 0$ if $u$ is the last vertex of $P_{i-1}$.
	By \Cref{lem:condition}, we have $\vec x^{|u}[L_i] \in Q^{(r-1)}_{G[L_i]}$ so the $i$th recursive call
	returning path $P_{i}$ is also well-defined.
	Furthermore,
\begin{align*}
	\Pr[v\in P] = \Pr[v\in P_{i}] &= \sum_{w\in V_{\sfrac{(i-1)\LH}{d}}} \Pr[w = u] \cdot \Pr[v\in P_{i} \mid w = u] \\
	&= \sum_{w\in V_{\sfrac{(i-1)\LH}{d}}} x_{w} \cdot \Pr[v\in P_{i} \mid w = u] \\
	&= \sum_{w\in V_{\sfrac{(i-1)\LH}{d}}} x_{w} \cdot x^{|w}_{v} = \sum_{w\in V_{\sfrac{(i-1)\LH}{d}}} x_{w,v} = x_{v} \mcomma
\end{align*}
where the last equality follows from~\eqref{lp:layer}.
	Note that for all $v\in V_{1+\sfrac{(i-1)\LH}{d}}$ with $(u,v)\notin A$, we have from~\eqref{lp:edge} that
	$x_{u,v} = 0$ and, in particular,
	$x^{|u}_{v} = 0$.
	Thus, by the induction hypothesis, we have that
	\begin{itemize}[nosep]
		\item $P_{i}$ cannot contain such a vertex and,
		\item $P_{1},\dotsc,P_d$ are all paths.
	\end{itemize}
	By the former, there is an arc from the last vertex of $P_{i-1}$ to the first vertex of $P_{i}$, which, together with the latter, implies that $P$ is indeed a path from $V_1$ to $V_{\LH}$.
\end{proof}

As discussed, our main goal is to use an LP-based approach to deal with problems that one can reduce to Submodular Orienteering with constraints.
Whereas the recursive randomized rounding algorithm described in \Cref{alg:recursive_rounding} shows how we can round a point $\vec x\in Q^{(r)}$ to a random layer-spanning path $P$, we still lack a way to make sure that $P$ has a large submodular value.
We address this next.

\subsection{Approximating submodular objectives by linear functions}
Our key ingredient to deal with submodular objectives is the following lemma.
It shows that there is a randomized procedure that not only returns a random path $P\sim \RRR(\vec x)$, but also returns a linear function $\ell$ that upper bounds $f$ on all layer-spanning paths and such that the expected value of $f(P)$ is related to $\E[\ell(\vec x)]$.
It is important to note that both $P$ and $\ell$ are random objects that depend on the random choices made by the algorithm, and that the expectation is taken over these random choices.

	For $I\subseteq V$ with $\Pr_{P\sim\RRR(\vec x)}[I\subseteq P] > 0$, let $\RRR^{|I}(\vec x)$ be the distribution $\RRR(\vec x)$ conditioned on $I$ being contained in the path.
\begin{lemma}\label{lem:round-or-cut}
	Let $d\in \Z_{\ge 2}$ and $r\in \ZZ$ so that $\LH = d^r$.
	Let $\vec x\in Q^{(r)}$ and $f: 2^V\rightarrow \RR$ be monotone submodular. 
	There is a randomized $n^{O(r)}$ time algorithm returning one random path $P_{t} \sim \RRR^{|t}(\vec x)$ for each $t\in V_{\LH}$ with $x_t > 0$, and in addition a linear function $\ell : \R^{V \choose \le r+1} \rightarrow \R$ such that 
		\begin{itemize}[nosep]
			\item The random path $\overline{P}$, which we set equal to $P_t$ with probability $x_t$ for each $t\in V_{\LH}$ with $x_t > 0$, is distributed as $\RRR(\vec x)$,
			\item $\displaystyle\E[f(\overline{P})]=\sum_{\substack{t\in V_{\LH}: \\ x_t > 0}} x_{t} \cdot \E[f(P_{t})] \ge \frac{\E[\ell(\vec x)]}{(d-1)(r+1)}$,
		\item $\displaystyle\ell(\1^{(r)}_{P'}) \ge f(P')$ for all layer-spanning paths $P'$ with probability~$1$, and
		\item $\ell(\vec y) = \sum_{I\in \binom{V}{\le r+1}} \ell_I \cdot y_I$ has coefficients $\ell_I\in [0, f(V)]$ $\;\forall I\in\binom{V}{\le r+1}$.
		\end{itemize}
\end{lemma}
\begin{proof}
	The first point is a consequence of $P_t \sim \RRR^{|t}(\vec x)$ and $\Pr_{P\sim\RRR(\vec x)}[t\in P] = x_t$, which
	holds by \Cref{lem:marginal-preserving}.
	For the next two points, we argue by induction over $\LH$.
	For simplicity, we assume that $x_v > 0$ for each $v\in V$. We discuss at the end of the proof
	how to handle vertices with $x_v = 0$.

	For $\LH = 1$, set $P_{t}$ to $(t)$ with probability $1$ for each $t\in V_1$.
	We return $\ell(\vec y) = \sum_{t\in V_1} y_{t} \cdot f(\{t\})$, which satisfies the
	claim because $\ell(\vec x) = \E_{P\sim \RRR(\vec x)}[f(P)]$ and $(d-1)(r+1)\ge 1$.

	Now assume that $\LH \ge d$ and, as in the algorithm, define $L_i = \{(i-1) \frac{\LH}{d}+1, (i-1)\frac{\LH}{d}+2,\dotsc,i\frac{\LH}{d}\}$ for each $i\in\{1,2,\dotsc,d\}$.
		We write $L_{\le i} = L_1 \cup \cdots \cup L_i$.

	\paragraph*{Subpaths and linear functions.}
	Using the induction hypothesis we will carefully construct the following random objects.
		For each $i\in\{1,2,\dotsc,d\}$ and $v\in V_{\sfrac{i\LH}{d}}$, we will construct a path $P^{|v}$ distributed as paths sampled from $\RRR^{|v}(\vec x)$ and then restricted to $G[L_{\le i}]$.
		The random paths of the statement are then defined as $P_{t} = P^{|t}$ for each $t\in V_{\LH}$.

		In the process of constructing the paths $P^{|v}$, we also construct, for each $i\in\{2,3\dotsc,d\}$, $u\in V_{\sfrac{(i-1)\LH}{d}}$, and $v\in V_{\sfrac{i\LH}{d}}$ with $x_{u,v} > 0$, a path $P^{|u,v}$ distributed as $\RRR^{|u,v}(\vec x)$ restricted to $G[L_i]$.
		Here, the submodular functions that we invoke the induction hypothesis with play a crucial role.

	First, via the induction hypothesis on $G[L_1]$ from $\vec x[L_1]$ with $f$,
	construct upper bound $\ell_{\bot}$ and paths $P^{|v}$ from $V_1$ to $v$ for each $v\in V_{\sfrac{\LH}{d}}$.

		Then for each $i\in\{2,3\dotsc,d\}$ (in that order): For each $w\in V_{\sfrac{(i-1)\LH}{d}}$ apply the induction hypothesis on $G[L_{i}]$ and $\vec x^{|w}[L_{i}]$ with $f(\, \cdot \mid P^{|w})$~\footnote{Here, $f(S \mid T) = f(S\cup T) - f(T)$ denotes the marginal gain of $S$ when added to $T$.} to obtain an upper bound $\ell^{|w}$ and paths $P^{|w,v}$ for each $v\in V_{\sfrac{i\LH}{d}}$ with $x_{w,v} > 0$ (equivalently, $x^{|w}_v > 0$).
		Now, conversely, for each $v\in V_{\sfrac{i\LH}{d}}$, we sample $u\in V_{\sfrac{(i-1)\LH}{d}}$
		with $\Pr[u = w] = x^{|v}_w$ for each $w\in V_{\sfrac{(i-1)\LH}{d}}$.
		Then we set $P^{|v}$ as the concatenation of $P^{|u}$ and $P^{|u,v}$.

		The fact that $P^{|v}$ is distributed as paths sampled from $\RRR^{|v}(\vec x)$ and then restricted to $G[L_{\leq i}]$ follows from
	\begin{align*}
		\Pr_{P \sim \RRR^{|v}(\vec x)}[w\in P]
		&= \Pr_{P \sim \RRR(\vec x)}[w\in P \mid v\in P]
		\\
		&= \frac{\Pr_{P \sim \RRR(\vec x)}[w\in P]
		\cdot \Pr_{P \sim \RRR(\vec x)}[v\in P \mid w\in P]}
		{ \Pr_{P \sim \RRR(\vec x)}[v\in P]} \\
		&= \frac{x_w \cdot x^{|w}_{v}}{x_v} = \frac{x_{w,v}}{x_v} = x^{|v}_w \mperiod
	\end{align*}
	For all linear functions defined above, we extend their domain to $\R^{\binom{V}{\le r+1}}$. Formally,
	we apply them to the projection to their respective domain.

	Let $i\in\{1,\dotsc,d-1\}$ and $w\in V_{\sfrac{i\LH}{d}}$.
	One complication that we face is that $\ell^{|w}$ is not an upper bound for $f(\, \cdot\, )$, but for $f(\, \cdot \mid P^{|w})$, and that it is linear in $\vec x^{|w}$ and not $\vec x$. To fix this, consider $\ell_{w}\colon \R^{\binom{V}{\le r+1}}\rightarrow \R$ defined by
	\begin{equation}
		\ell_{w}(\vec y) = y_w \cdot f(P^{|w}) + \ell^{|w}(\vec y^{+w}) \mcomma
	\end{equation}
	where $y^{+w}_I = y_{I\cup\{w\}}$ for all $I\in \binom{V}{\le r}$.
	Then
	every path $P'$ from $w$ to $V_{\sfrac{(i+1)\LH}{d}}$ satisfies
	\begin{equation}\label{eq:upper-bound}
		\ell_{w}(\1^{(r)}_{P'}) = f(P^{|w}) + \ell^{|w}(\1^{(r)}_{P'}) \ge f(P^{|w}) + f(P' \mid P^{|w}) \ge f(P') \mperiod
	\end{equation}
	Adding all these linear functions will yield our upper bound for $f$. We define
	\begin{equation*}
		\ell(\vec y) = \ell_{\bot}(\vec y) + \sum_{i=1}^{d-1} \sum_{w\in V_{\sfrac{i\LH}{d}}} \ell_{w}(\vec y) \mperiod
	\end{equation*}
	Next, we show that $\ell$ has the claimed properties.
	\paragraph*{Expected values of paths.}
	We will now lower bound
	$\sum_{t\in V_{\LH}} x_t \cdot \E[f(P^{|t})]$.
	Let $i\in\{2,3,\dotsc,d\}$. Then
	\begin{align}
		\sum_{v\in V_{\sfrac{i\LH}{d}}} &x_v \cdot \E[f(P^{|v})] - \left(1-\frac{1}{(d-1)r}\right)\sum_{w\in V_{\sfrac{(i-1)\LH}{d}}} x_{w} \cdot \E[f(P^{|w})]\notag\\
		 &= \sum_{v\in V_{\sfrac{i\LH}{d}}} x_v \sum_{\substack{w\in V_{\sfrac{(i-1)\LH}{d}} \\ x_{\{w,v\}} > 0}} x^{|v}_w \cdot (\E[f(P^{|w})] + \E[f(P^{|w,v} \mid P^{|w})]) \notag\\ 
		 &\quad - \left(1-\frac{1}{(d-1)r}\right)\sum_{w\in V_{\sfrac{(i-1)\LH}{d}}} x_{w} \cdot \E[f(P^{|w})] \notag\\
		 &= \sum_{w\in V_{\sfrac{(i-1)\LH}{d}}} x_w \sum_{\substack{v\in V_{\sfrac{i\LH}{d}} \\ x_{\{w,v\}} > 0}} x^{|w}_v \cdot (\E[f(P^{|w})] + \E[f(P^{|w,v} \mid P^{|w})]) \notag\\ 
		 &\quad - \left(1-\frac{1}{(d-1)r}\right)\sum_{w\in V_{\sfrac{(i-1)\LH}{d}}} x_{w} \cdot \E[f(P^{|w})] \notag\\
		 &= \sum_{w\in V_{\sfrac{(i-1)\LH}{d}}} x_w \bigg(\frac{\E[f(P^{|w})]}{(d-1)r} + \sum_{\substack{v\in V_{\sfrac{i\LH}{d}} \notag\\ x_{\{w,v\}} > 0}} x^{|w}_v \cdot \E[f(P^{|w,v} \mid P^{|w})]\bigg) \notag\\
		 &\ge \sum_{w\in V_{\sfrac{(i-1)\LH}{d}}} \frac{x_{w}}{(d-1)r} \left(\E[f(P^{|w})] + \E[\ell^{|w}(\vec x^{|w})]\right) \notag\\
		 &= \frac{1}{(d-1)r} \E\left[\sum_{w\in V_{\sfrac{(i-1)\LH}{d}}} \ell_{w}(\vec x)\right] \mcomma\label{eq:expansionEfP}
	\end{align}
	where the inequality follows from the induction hypothesis applied to the construction of $\ell^{|w}$ and the paths $\{P^{|w,v}\}_{v\in V_{\sfrac{iH}{d}}}$. (More precisely, we use the second point of \cref{lem:round-or-cut}.)

	By summing~\eqref{eq:expansionEfP} over all $i\in\{2,3,\dotsc,d\}$, we obtain
	\begin{align*}
		\sum_{t\in V_{\LH}} x_t \cdot \E[f(P^{|t})] &\ge \sum_{v\in V_{\sfrac{\LH}{d}}} x_v \cdot \E[f(P^{|v})] \\
		&\quad + \frac{1}{(d-1)r}\sum_{i=1}^{d-1} \left[\E\left[\sum_{w\in V_{\sfrac{i\LH}{d}}} \ell_{w}(\vec x)\right] - \sum_{w\in V_{\sfrac{i\LH}{d}}} x_w \cdot \E[f(P^{|w})]\right] \\
		&\ge \frac{1}{(d-1)r}\E[\ell_{\bot}(\vec x)] - \frac{1}{r}\sum_{t\in V_{\LH}} x_t \cdot \E[f(P^{|t})] + \frac{1}{(d-1)r}\sum_{i=1}^{d-1} \E\left[\sum_{w\in {V_{\sfrac{i\LH}{d}}}} \ell_{w}(\vec x)\right] \mperiod
	\end{align*}
	In the second inequality we use that $\sum_{w\in V_{\sfrac{i\LH}{d}}} x_w \cdot \E[f(P^{|w})] \le \sum_{t\in V_{\LH}} x_t \cdot \E[f(P^{|t})]$ for all $i$.
	By moving the terms related to $P^{|t}$ to the left and multiplying with $\frac{r}{r+1}$, it follows that
	\begin{align*}
		\sum_{t\in V_{\LH}} x_t \cdot \E[f(P^{|t})] &\ge \frac{1}{(d-1)(r+1)}\E[\ell_{\bot}(\vec x)] + \frac{1}{(d-1)(r+1)} \sum_{i=1}^{d-1} \E\left[\sum_{w\in V_{\sfrac{i\LH}{d}}} \ell_{w}(\vec x)\right] \\
		&= \frac{1}{(d-1)(r+1)} \E[\ell(\vec x)] \mperiod
	\end{align*}

	\paragraph*{Upper bound on any layer-spanning path.} It remains
	to show that $\ell$ upper bounds $f(P')$ for any layer-spanning path $P'$.
	For each $i\in\{1,2,\dotsc,d\}$,
	let $u_i$ be the unique vertex in the intersection of $P'$ and $V_{\sfrac{i\LH}{d}}$
	and let $P'_{i}$ be the restriction of $P'$ to $G[L_i]$.
	By induction and because $P'_{1}$ is layer-spanning in $G[L_1]$, we have that $\ell_{\bot}(\1^{(r)}_{P'}) \ge f(P'_{1})$.
	By~\eqref{eq:upper-bound}, we also have $\ell_{u_{i}}(\1^{(r)}_{P'}) \ge f(P'_{i+1})$ for each $i\in\{1,\dotsc,d-1\}$.
	We conclude that
	\begin{align*}
		\ell(\1^{(r)}_{P'}) &= \ell_{\bot}(\1^{(r)}_{P'}) + \sum_{i=1}^{d-1} \sum_{w\in V_{\sfrac{i\LH}{d}}} \ell_w(\1^{(r)}_{P'}) \\
		&= \ell_{\bot}(\1^{(r)}_{P'}) + \sum_{i=1}^{d-1} \ell_{u_{i}}(\1^{(r)}_{P'}) \ge f(P'_1) + \cdots + f(P'_d) \ge f(P') \mperiod
	\end{align*}
	\paragraph*{Concluding remarks.}
	Note that by construction, the coefficients of $\ell$ are conic combinations of polynomially many coefficients from recursive constructions with $r-1$, but we have not discussed their magnitude.
	We can replace each coefficient $\ell_I$ by $\min\{f(V),\ell_I\}$, since this maintains
	the other properties. This gives the upper bound on the coefficients.

	In the case that $x_v = 0$ for some $v\in V$, we can perform the construction above on the graph
	restricted to vertices with non-zero variables. Then, we add to the resulting function $f(V) \cdot y_v$
	for every $v\in V$ with $x_v = 0$. This does not affect the lower bounds on the expected value of
	$\RRR(\vec{x})$, but ensures that every path that uses vertices with zero $x$-value satisfies the lower bound on $\ell$.
\end{proof}

To highlight the main ideas of how \Cref{lem:round-or-cut} can be leveraged, we start with
a quick and instructive warm-up, and we will later exploit that our approach is LP-based by adding additional constraints to the linear description to deal with further applications.

\subsection{Round-or-cut algorithm for Submodular Orienteering}
\label{sec:main-orienteering}
We show how to use \Cref{lem:round-or-cut} to obtain an $O(dr)$-approximation for a decision variant of Submodular Orienteering on layered graphs, without a length bound and with the number of operations of the algorithm depending polynomially
on $\langle f \rangle$, the largest encoding size of a function value $f(S)$ with $S\subseteq V$.
Both of these limitations will be removed in \Cref{sec:orienteering-full}, where we derive \Cref{thm:main-orienteering}.

In the decision version,
for a given bound $T\in \mathbb{R}_{\geq 0}$, we either output a solution of value at least $\frac{T}{4(d-1)(r+1)}$ or determine that $f(\opt) < T$, where $\opt$ is an optimal solution.

We highlight that, in the interest of simplicity and clarity, we did not make any effort to optimize the constant factor in front of $dr$ in our analysis, which can easily be improved.

\Cref{alg:basic_round-or-cut} describes our round-or-cut procedure for Submodular Orienteering.
For any $T\in \mathbb{R}_{\geq 0}$, we define the truncated submodular function $f^T(S) \coloneqq \min\{f(S), T\}$ for all $S\subseteq V$.

\begin{algorithm2e}
\SetKwFor{Whenever}{whenever}{do}{end}
\SetKwFor{Loop}{loop}{}{end}
\DontPrintSemicolon
\caption{Round-or-cut for Submodular Orienteering}\label{alg:basic_round-or-cut}
\KwIn{Layered orienteering graph $G=(V,A)$ with $V=V_1 \cup \cdots \cup V_{\LH}$, $\LH = d^r$, monotone submodular function $f\colon 2^V \to \mathbb{R}_{\geq 0}$, and $T \in \mathbb{R}_{\geq 0}$.}
	\KwOut{layer-spanning path $P$ with $f(P)\geq \frac{T}{4(d-1)(r+1)}$ or determine that $f(\opt) < T$.}

Initialize and start ellipsoid method to find point in a polytope in $[0,1]^{\binom{V}{\leq r+1}}$ given by a separation oracle.\;

\Whenever{\label{algline:sep_call}Ellipsoid calls separation oracle with point $\vec x\in [0,1]^{\binom{V}{\leq r+1}}$}{
\uIf{$\vec x \not \in Q^{(r)}$}{pass a violated constraint of $Q^{(r)}$ to ellipsoid and continue with ellipsoid on line~\ref{algline:sep_call}.}\label{algline:checkQh}
\Else{
	\Loop{}{
		Use \Cref{lem:round-or-cut} with $f^T$ to get random layer-spanning path $P$ (corresponding to $\overline{P}$ in \Cref{lem:round-or-cut}) and random linear function $\ell$.\;
		\If{$f(P) \geq \frac{T}{4(d-1)(r+1)}$}{\Return{$P$}.}
		\If{$\ell(\vec x) < T$}{
			Pass cutting plane $\ell(\vec x) \geq T$ to ellipsoid and continue with ellipsoid on line~\ref{algline:sep_call}.}\label{algline:ellCut}
	}\label{algline:loop_basic_round-or-cut}
}
}
	\If{ellipsoid terminates because no vector satisfies generated cutting planes}{\Return{``$f(\opt) < T$''}.}

\end{algorithm2e}

We now briefly discuss correctness and running time of \Cref{alg:basic_round-or-cut}.
First note that there are only two ways for the algorithm to terminate: either it outputs a path $P$ with $f(P) \geq \frac{T}{4(d-1)(r+1)}$, or it outputs ``$f(\opt) < T$''.
Hence, if $f(\opt) < T$, then if the algorithm terminates, it returns a correct result.

Now consider the case that $T \le f(\opt)$.
Note that cutting planes generated in \Cref{alg:basic_round-or-cut} are valid for the vector $\1^{(r)}_{\opt}\in Q^{(r)}$.
This is clearly true for the constraints of $Q^{(r)}$, which get checked in line~\ref{algline:checkQh} of \Cref{alg:basic_round-or-cut}.
Moreover, a cutting plane $\ell(\vec x) \geq T$ generated in line~\ref{algline:ellCut} is valid because $\ell(\1^{(r)}_{\opt}) \geq f^T(\opt) = \min\{f(\opt), T\} = T$, where the first inequality follows from \Cref{lem:round-or-cut}.
Thus, also in this case, if the algorithm terminates, it returns a correct result.

It remains to discuss the running time of \Cref{alg:basic_round-or-cut}.
Note that checking whether $\vec x\in Q^{(r)}$ can be done in time $n^{O(r)}$ by checking all constraints of $Q^{(r)}$.
All other single operations in \Cref{alg:basic_round-or-cut} can be done in polynomial time.
The number of iterations of the ellipsoid method is polynomial in the dimension of the space and the bit complexity of the cuts.
The bit complexity of the constraints in $Q^{(r)}$ is polynomial, and the bit complexity of the cuts $\ell(\vec x) \geq T$ is also polynomial in $n^r$ and the maximum bit complexity of a function value of $f$.

Hence, to obtain an expected $n^{O(r)} \langle f \rangle^{O(1)}$ time algorithm, it suffices to show that the expected number of iterations of the inner loop in \Cref{alg:basic_round-or-cut} is polynomial, which we show next.

\begin{lemma}
	In every iteration of the inner loop in \Cref{alg:basic_round-or-cut}, the probability that the algorithm either outputs a path $P$ with $f(P) \geq \frac{T}{4(d-1)(r+1)}$ or generates a cutting plane $\ell(\vec x) \geq T$ is at least $\Omega(\frac{1}{dr})$.
\end{lemma}
\begin{proof}
	Let $\vec x\in Q^{(r)}$ be the point that is currently passed to the inner loop in \Cref{alg:basic_round-or-cut}, and let $P$ and $\ell$ be the random path and linear function obtained from \Cref{lem:round-or-cut} with respect to $\vec x$ and $f^T$.

	If $\E[f^T(P)] < \frac{T}{2(d-1)(r+1)}$, then \Cref{lem:round-or-cut} implies 
	\begin{equation*}
		\E[\ell(\vec x)] \le (d-1)(r+1) \cdot \E[f^T(P)] < \frac{T}{2} \mperiod
  \end{equation*}
	Thus, by Markov's inequality we get $\Pr[\ell(\vec x) < T] \ge \sfrac{1}{2}$, and the algorithm generates a cutting plane $\ell(\vec x) \geq T$ with probability at least $\sfrac{1}{2}$.

	Suppose now that $\E[f^T(P)] \geq \frac{T}{2(d-1)(r+1)}$, which implies
	\begin{align*}
		\frac{T}{2(d-1)(r+1)} &\le \E[f^T(P)] \\
		&\le \Pr\left[f^T(P) \ge \frac{T}{4(d-1)(r+1)}\right] \cdot T \\
		&\quad + \left(1 - \Pr\left[f^T(P) \ge \frac{T}{4(d-1)(r+1)}\right]\right) \cdot \frac{T}{4(d-1)(r+1)} \mperiod
	\end{align*}
	By rearranging we get
	\begin{equation*}
		\Pr\left[f^T(P) \ge \frac{T}{4(d-1)(r+1)}\right] \geq \frac{1}{4(d-1)(r+1) - 1} \mcomma
	\end{equation*}
	and the statement follows from $f(P) \geq f^T(P)$.
\end{proof}

Hence, in summary, we obtained an expected $n^{O(r)} \langle f \rangle^{O(1)}$ time $O(dr)$-approximation for the decision version of Submodular Orienteering (without a length bound) by using \Cref{lem:round-or-cut} in a round-or-cut procedure.

Consequently, \Cref{alg:basic_round-or-cut} returns an $O(dr)$-approximation for Submodular Orienteering if run with a $T$ such that $T \le f(\opt)$ and $T = \Omega(f(\opt))$.
Such a $T$, even with the property that $f(\opt) \geq T \geq (1-\varepsilon)f(\opt)$, can typically be found easily through binary search using \Cref{alg:basic_round-or-cut}.
See \Cref{sec:orienteering-full} for a more detailed discussion of this point in the more general setting with additional constraints.
This leads to approximation ratios of $O(\log n)$ and $O(\sfrac{n^\epsilon}{\epsilon})$ in quasi-polynomial time and polynomial time, respectively.
The $\sfrac{1}{\epsilon}$ factor in the latter can be removed by rescaling $\epsilon$.

\subsection{Round-or-cut with linear constraints}
\label{sec:orienteering-full}

As discussed, a main benefit of our LP-based approach is that we can easily add constraints
to the linear description and use them in a round-or-cut procedure to obtain good approximation guarantees for submodular maximization.
To this end, consider an arbitrary family of additional constraints we want to add to $Q^{(r)}$.
Let us denote these constraints by $D\vec y \leq \vec{b}$ for some matrix $D\in \Z^{m\times \binom{V}{\le r + 1}}$ and vector $\vec{b}\in \Z^m$, where $\vec y\in \R^{\binom{V}{\leq r+1}}$. We denote by $\langle D, \vec b \rangle$
the encoding size of $D$ and $\vec b$, which is upper bounded by $m\cdot |\binom{V}{\leq r+1}| \cdot (1 + \lceil \log(1 + \|D\|_{\infty}) + \log(1 + \|\vec b\|_{\infty}) \rceil)$.\footnote{We use $\log(\cdot)$ for base-$2$ logarithm and $\ln(\cdot)$ for the natural logarithm. In most places, however, this distinction is not important.}
If the constraints are rational (but non-integer), we assume they are first normalized by multiplying with the lowest common denominator. 

We need to adapt the round-or-cut procedure to deal with additional constraints.
It is typically easy to separate over additional constraints.
However, if we use a procedure akin to \Cref{alg:basic_round-or-cut}, we have to be aware that this procedure returns a path $P$, rounded from some point $\vec{x}$, only when a certain condition is satisfied (in this case $f(P) \geq \frac{T}{4(d-1)(r+1)}$).
Otherwise, we may repeat the loop.
This may introduce biases in the distribution of $P$ that are hard to control.
Furthermore, we need to specify what we compare our solution against.

\paragraph*{Optimum under linear constraints.}
Previously, $\opt$ was a layer-spanning path of maximal function value.
In the new setting, it would be natural to restrict $\opt$ to a path satisfying the linear constraints.
However, a weaker restriction will suffice for our analysis and this broadens the range of applications.
Specifically, we compare against $\E_{\opt\sim\mathcal O} [f(\opt)]$ for a distribution $\mathcal O$
over layer-spanning paths
that satisfies the linear constraints in expectation, that is,
$\E_{\opt\sim\mathcal O}[D \, \1^{(r)}_{\opt}] \le \vec b$. We assume without loss of generality
that $\mathcal O$ is the distribution
with maximal expected function value and we omit $\opt\sim\mathcal O$ in $\E[\;\cdot\;]$ if it is
clear from the context.

\paragraph*{Main statement.}
Our goal is to adapt the round-or-cut procedure to obtain the following.\begin{lemma}\label{lem:getXWithConstraints}
	Let $G=(V,A)$ be a layered graph with $V = V_1 \cup \cdots \cup V_{\LH}$, $\LH = d^r$, let $f\colon 2^V\to \RR$ be monotone submodular, let $D\vec{y} \leq \vec{b}$ be constraints for $\vec{y}\in\R^{\binom{V}{\leq r+1}}$ with $Q^{(r)}\cap \{\vec{y}\in\R^{\binom{V}{\leq r+1}} \colon D\vec{y}\leq \vec{b}\}\neq \emptyset$, and let $\varepsilon>0$.

	There is a randomized $n^{O(r)} \cdot \langle D, \vec b \rangle^{O(1)}$ time procedure that returns an $\vec{x}\in Q^{(r)}$ with $D\vec{x} \leq \vec{b}$ and, with probability $1-n^{-\frac{1}{\varepsilon}}$, the returned $\vec{x}$ satisfies
	\begin{equation*}
		\E_{P\sim \RRR(\vec x)}[f(P)] \ge \frac{1-\eps}{(d-1)(r+1)}\cdot\E_{\opt\sim\mathcal O}[f(\opt)] \mcomma
	\end{equation*}
	for any distribution $\mathcal O$ over layer-spanning paths that satisfies $\E_{\opt\sim \mathcal O}[D \, \1^{(r)}_{\opt}] \le \vec b$.
\end{lemma}
There is a subtle aspect about this lemma that is worth pointing out: 
$\E_{P\sim \RRR(\vec x)}[D\, \1^{(r)}_P] \le \vec{b}$ may not hold. This is because $\RRR(\vec x)$ is not necessarily
marginal preserving on variables $x_I$ with $|I| > 1$.
Thus, the distribution $\RRR(\vec x)$ is in this regard less constrained than $\mathcal O$ (though it is also more constrained in other regards).
In particular, it is possible that 
$\E_{P\sim \RRR(\vec x)}[f(P)] > \E_{\opt\sim \mathcal O}[f(\opt)]$ for some $\vec x$ with $D\vec x\le \vec b$.

Before proving the lemma, we first show how it implies \Cref{thm:main-orienteering}.
\begin{proof}[Proof of \Cref{thm:main-orienteering} assuming \Cref{lem:getXWithConstraints}]
Suppose we are given an instance of Submodular Orienteering. 
We reduce from a general, not necessarily layered, graph to a layered one as follows.
	Let $H \ge n+2$ to be specified later. We consider the following vertices:
	for each $i\in\{2,\dotsc,H-1\}$ and each $u,v\in V$, where $v$ is reachable from $u$ in the original graph (in particular, if $u=v$),
	we let $(u,v,i)$ be a vertex in $V_i$. Furthermore, we set $V_1 = \{(s,s,1)\}$ and $V_H = \{(t,t,H)\}$.
	There is an arc from $(u,v,i)\in V_i$ to $(u',v',i+1)\in V_{i+1}$ if $v = u'$. With each vertex $(u,v,i)$ we associate a length, which is the length of the shortest path from $u$ to $v$ in the original graph.
	The submodular function is extended to $f'(U) \coloneqq f(\{v \mid \exists i\in\{1,\dotsc,H\}, u\in V \text{ with } (u,v,i) \in U\})$, which is again submodular.
This reduction is approximation-preserving: Each walk $W$ in the original graph
	can be transformed into a layer-spanning path of the same value and at most the same length in the layered graph.
	For this, let $v_1,\dotsc,v_\ell$ for some $\ell \le n$ be the sequence of distinct vertices visited in $W$, in the order they first appear in $W$. Then $s = v_1$ and a layer-spanning path satisfying this is
	\begin{equation*}
		(v_1,v_1,1),(v_1,v_2,2),(v_2,v_3,3),\dotsc,(v_\ell,t,\ell+1),(t,t,\ell+2),\dotsc,(t,t,H) \mperiod
	\end{equation*}
	On the other hand, any layer-spanning path can be translated into an $s$-$t$ walk of the same length and at least the same value
	by replacing each $(u,v,i)$ by a shortest path from $u$ to $v$ and concatenating all paths.\footnote{
	An alternative reduction to a layered graph would be to add $(u,v,i)$ to $V_i$ only for vertices $u,v\in V$ if there is a $u$-$v$ arc in the Submodular Orienteering instance (instead of just a $u$-$v$ path). 
	In this case, the vertex $(u,v,i)$ would be assigned the length of the corresponding $u$-$v$ arc.
	One (slight) disadvantage of this approach is that $\Omega(n^2)$ layers are necessary instead of only $O(n)$, because the number of arcs in an optimal Submodular Orienteering solution can be in the order of $n^2$.
	(It is $O(n^2)$ because an optimal solution can be constructed so newly visited vertices are connected by shortest paths.)
	}

	We use $D,\vec b$ to encode the length constraint. More precisely, $D$ contains a single
	row where the coefficient of each $x_{\{(u,v,i)\}}$ is set to the shortest $u$-$v$ path length in the original graph.
	All other coefficients, i.e., those corresponding to combinations of variables, are set to zero.
	The single component of the right-hand side $\vec b$ is the budget $C$, for which we can
	assume $C\le H n\cdot \max_{a\in A} w(a)$, since the
	length of $H$ many shortest paths cannot exceed the right-hand side.
	Thus, $\langle D,\vec b\rangle \le (\langle w \rangle + n)^{O(1)}$.
	We compute $\vec x$ using \Cref{lem:getXWithConstraints} and output $P\sim\RRR(\vec x)$.
	Then the expected length of $P$ is at most $C$ because of \Cref{lem:marginal-preserving}.
	For an optimal solution $\opt$, let $\mathcal O$ be the distribution that picks $\opt$ with probability $1$.
	Then
	\begin{equation*}
		\E[f(P)] \ge (1 - n^{-1 / \eps}) \frac{1-\eps}{(d-1)(r+1)} f(\opt) \mperiod
	\end{equation*}
	By setting $d = 2$ and $r = \Theta(\log n)$, or $d = n^{\Theta(\epsilon)}$ and $r = \Theta(\sfrac{1}{\eps})$, we get the claimed trade-offs.
\end{proof}
The rest of the section is dedicated to proving \Cref{lem:getXWithConstraints}.

\paragraph*{Initial bound.}
We show how to compute a solution $\vec{x}$ with a weak initial guarantee.
This allows us later to restrict the range of the binary search for the threshold $T$, which will play in our new algorithm an analogous role to the threshold $T$ in \Cref{alg:basic_round-or-cut}.
\begin{lemma}\label{lem:min-probability}
	In time $n^{O(r)} \cdot \langle D, \vec b \rangle^{O(1)}$
	we can compute $U\subseteq V$, $\vec x\in Q^{(r)}$, and $c_D \le 2^{n^{O(r)} \langle D, \vec b \rangle^{O(1)}}$ 
	such that
	$D\vec x\le \vec b$ and
	for every $u\in V\setminus U$ we have $\Pr_{\opt\sim\mathcal O}[u\in \opt] = 0$.
	Furthermore,
	\begin{equation*}
		\E_{P\sim \RRR(\vec x)}[f(P)] \ge \frac{1}{n c_D} \cdot f(U) \mperiod
	\end{equation*}
\end{lemma}
\begin{proof}
	Let $u \in U$ if and only if there 
	exists a solution $\vec x\in Q^{(r)}$
	with $D\vec x\le \vec b$ and $x_{u} > 0$.
	The set $U$ can be computed by solving a linear program for each vertex in the required time.
	Note that every $v\in V$ with $\Pr_{\opt\sim\mathcal O}[v\in \opt] > 0$ must be in $U$, since one can take $\vec x = \E_{\opt\sim\mathcal O}[\1^{(r)}_\opt]$.
	Let $\vec x$ be a vertex solution with the properties as above for $u\in U$ maximizing $f(\{u\})$, which can
	be computed the same way.

	By Cramer's rule, the denominator of $x_{u}$ is the determinant of a submatrix of the constraint matrix of $Q^{(r)} \cap \{\vec x : D\vec x \le \vec b\}$, which is bounded by some $c_D\le 2^{n^{O(r)} \langle D, \vec b \rangle^{O(1)}}$.
	It follows that
	\begin{equation*}
		\E_{P\sim \RRR(\vec x)}[f(P)] \ge x_{u} f(\{u\}) \ge \frac{1}{n c_D} \cdot f(U) \mperiod
	\end{equation*}
	The first inequality follows from $\Pr[u\in P] = x_u$ (by \Cref{lem:marginal-preserving}) and monotonicity of $f$.
\end{proof}

\paragraph*{Avoiding bit complexity dependence.}
The next lemma will be used to round the linear functions obtained from \Cref{lem:round-or-cut} in order to
avoid cuts with a high bit complexity in the ellipsoid method
and to help us avoid unnecessary dependence of the running time
on the bit complexity of~$f$.
\begin{lemma}\label{lem:round}
	Let $c_D$ be as in \Cref{lem:min-probability}.
	Let $\ell(\vec y) = \sum_{I\in\binom{V}{\le r+1}} \ell_I y_I$ be a linear function over $\vec y\in \R^{\binom{V}{\le r+1}}$ with $\ell_I \in [0,f(V)]$ for $I\in \binom{V}{\leq r+1}$.
	Let $T\in\left[\frac{f(V)}{n c_D}, f(V)\right]$.

	We can compute in time $n^{O(r)}$ another linear function $\overline\ell(\vec y) = \sum_{I\in\binom{V}{\le r+1}} \overline\ell_I y_I$ and some $\overline T\in \ZZ$ such that all coefficients $\overline\ell_I$ and also $\overline T$ are non-negative integers bounded by $2 c_D^2 n^{r+3}/\eps + 1$ satisfying
	\begin{itemize}[nosep]
		\item $\overline\ell(\vec y) < \overline T \Rightarrow \ell(\vec y) < T$ for all $\vec y \in [0,1]^{\binom{V}{\le r+1}}$, and
		\item $\ell(\vec y) < (1 - \eps) T \Rightarrow \overline\ell(\vec y) < \overline T$ for all $\vec y \in [0,1]^{\binom{V}{\le r+1}}$.
	\end{itemize}
\end{lemma}
\begin{proof}
	We set
	\begin{equation*}
		\overline T \coloneqq \left\lceil\; \frac{c_D n T}{\eps f(V)} \cdot \left|\binom{V}{\le r+1}\right| \;\right\rceil \in \left[\; \frac{1}{\eps} \cdot \left|\binom{V}{\le r+1}\right| , \frac{2 c_D n}{\eps} \cdot \left|\binom{V}{\le r+1}\right| \;\right] \mperiod
	\end{equation*}
	Furthermore, for every $I\in \binom{V}{\le r+1}$ we set
	\begin{equation*}
		\overline\ell_I \coloneqq \left\lceil\ell_I \cdot \frac{\overline{T}}{T} \right\rceil \le f(V) \cdot \frac{\overline{T}}{T} + 1 \le \overline T c_D n + 1 \mperiod
	\end{equation*}
	We now verify the two claimed properties of $\overline \ell$ and $\overline T$.
	Let $\vec y\in [0,1]^{\binom{V}{\le r+1}}$. If $\overline \ell(\vec y) < \overline T$ then
	\begin{equation*}
		\ell(\vec y) \le \overline \ell(\vec y) \cdot \frac{T}{\overline{T}}< T \mperiod
	\end{equation*}
	If, on the other hand, $\ell(\vec y) < (1-\eps)T$, then
	\begin{equation*}
		\overline\ell(\vec y) \le \ell(\vec y) \cdot \frac{\overline{T}}{T} + \left|\binom{V}{\le r+1}\right|
		< (1-\eps)\overline{T} + \eps\overline{T} = \overline T \mperiod \qedhere
	\end{equation*}
\end{proof}
\medskip

\noindent
\Cref{lem:getXWithConstraints} can be obtained by a modification of \Cref{alg:basic_round-or-cut}, which is described in \Cref{alg:constrained_round-or-cut}. We combine this algorithm with a few additional steps. 
Then we invoke \Cref{lem:min-probability}.
If $U \subsetneq V$, we can
remove all vertices in $V\setminus U$ from the instance, since they are not relevant for $\opt$, and restart
the algorithm on the smaller instance.
Thus, assume without loss of generality $V = U$.
We perform a binary search over $T\in \left[ \frac{f(V)}{n c_D}, f(V) \right]$.
We stop the search when the ratio between upper and lower bound is at most $\frac{1}{1 - \eps}$. 
If \Cref{alg:constrained_round-or-cut} fails in every invocation, we return the solution from \Cref{lem:min-probability}.

In \Cref{alg:constrained_round-or-cut}, $\kappa$ is chosen to be a sufficiently large constant so that $(1-\varepsilon)^{\kappa \cdot r\ln(n+\langle D, \vec b \rangle)} \leq \langle D, \vec b \rangle^{-c} n^{-rc - \frac{1}{\varepsilon}}$, where $n^{cr} \cdot \langle D, \vec b \rangle^c$ is an upper bound for the number of iterations of the binary search times the number of iterations needed by the ellipsoid method to finish when run with constraints of $Q^{(r)}$ and $D\vec{x} \leq \vec{b}$, and with further constraints with coefficients bounded by $2c^2_D n^{r+3}/\eps + 1$,
the term from \Cref{lem:round}.
This is for example achieved by $\kappa \coloneqq \frac{1}{\varepsilon} \cdot (2c + \frac{1}{\varepsilon})$.

\begin{algorithm2e}
\SetKwFor{Whenever}{whenever}{do}{end}
\SetKwFor{Loop}{loop}{}{end}
\SetKwFor{RepeatTimes}{repeat}{times}{end}
\DontPrintSemicolon
\caption{Round-or-cut for constrained Submodular Orienteering}\label{alg:constrained_round-or-cut}
\KwIn{Layered orienteering graph $G=(V,A)$ with $V=V_1 \cup \cdots \cup V_{\LH}$, $\LH = d^r$, monotone submodular function $f\colon 2^V \to \mathbb{R}_{\geq 0}$, $T \in \mathbb{R}_{\geq 0}$, and constraints $D\vec{y}\leq \vec{b}$
for $\vec{y}\in[0,1]^{\binom{V}{\leq r+1}}$ with $Q^{(r)}\cap \{\vec{y}\in\R^{\binom{V}{\leq r+1}} \colon D\vec{y}\leq \vec{b}\}\neq \emptyset$.}
	\KwOut{$\vec{x}\in Q^{(r)}$ with $D\vec x \le \vec b$ such that
	$\E_{P\sim\RRR(\vec x)}[f(P)] \ge \frac{(1-\eps)^2}{(d-1)(r+1)} \cdot T$
	or determine that $\E[f(\opt)]<T$.}

Initialize and start ellipsoid method to find point in a polytope in $[0,1]^{\binom{V}{\leq r+1}}$ given by a separation oracle.\;

\Whenever{\label{algline:sep_call_constrained}Ellipsoid calls separation oracle with point $\vec x\in [0,1]^{\binom{V}{\leq r+1}}$}{
\uIf{$\vec x \not \in Q^{(r)}$ or $D\vec{x} \not\leq \vec{b}$}{pass a violated constraint of $Q^{(r)}$ or $D\vec{x}\leq \vec{b}$ to ellipsoid and continue on line~\ref{algline:sep_call_constrained}.}\label{algline:checkQh_constrained}
\Else{
	\RepeatTimes{\label{algline:repeat_kappa_logn}$\kappa \cdot r \lceil \ln(n + \langle D, \vec b \rangle)\rceil$}{
		Use \Cref{lem:round-or-cut} with $f$ to get random layer-spanning path $P$ and random linear function $\ell$.\;
		Use \Cref{lem:round} with $\ell$ and $T$ to get $\overline\ell$, $\overline T$.\;
		\If{$\overline\ell(\vec x) < \overline T$}{pass cutting plane $\overline\ell(\vec x) \geq \overline T$ to ellipsoid and continue with ellipsoid on line~\ref{algline:sep_call_constrained}.}\label{algline:ellCutConstrained}
	}
	\Return{$\vec{x}$}
}
}

	\If{ellipsoid terminates because no vector satisfies generated cutting planes}{\Return{``$\E[f(\opt)] < T$''}.}

\end{algorithm2e}

We now show that the procedure indeed fulfills the properties stated in \Cref{lem:getXWithConstraints}.

\begin{proof}[Proof of \Cref{lem:getXWithConstraints}]
	We first discuss correctness of \Cref{alg:constrained_round-or-cut}.
	If the algorithm outputs that ``$\E[f(\opt)] < T$'', then this is correct due to the following.
Let $\vec x = \E[\1^{(r)}_{\opt}]$. Then $\vec x \in Q^{(r)}$ and $D \vec x \le \vec b$.
	Since the generated cutting planes define an infeasible set of constraints,
	there must be a generated cut with $\overline\ell(\vec x) < \overline T$.
	This cut was derived from some $\ell$ and $T$ using \Cref{lem:round}. By the
	lemma, it follows that $\ell(\vec x) < T$.
	Since $\ell$ was constructed by \Cref{lem:round-or-cut},
	for any layer-spanning path $P$
	we have $f(P) \le \ell(\1^{(r)}_P)$. Thus,
	$\E[f(\opt)] \le \E[\ell(\1^{(r)}_{\opt})] = \ell(\vec x) < T$
	and the conclusion is correct.

	If the algorithm does not return ``$\E[f(\opt)] < T$'', then it must return a point $\vec{x}\in Q^{(r)}$ with $D\vec{x} \leq \vec{b}$. Indeed, the ellipsoid method runs in polynomial time and therefore cannot stay forever inside the loop starting at line~\ref{algline:sep_call_constrained}.
Consider now any moment in the execution of \Cref{alg:constrained_round-or-cut} when it is at line~\ref{algline:repeat_kappa_logn}.
Let $\vec{x}\in Q^{(r)}$ be the point that is currently passed to the repeat loop, and let $P$ and $\ell$ be the random path and random linear function obtained from \Cref{lem:round-or-cut} with respect to $\vec x$.
	If $\Pr[\ell(\vec{x})<(1-\eps) T] \geq \varepsilon$, then since $\ell(\vec{x})<(1 - \eps)T$ implies $\overline\ell(\vec{x})< \overline T$, the algorithm will generate in each repetition of the repeat loop a cutting plane with probability at least $\varepsilon$.
	Hence, the probability that no cutting plane is generated in $\kappa \cdot r \lceil\ln(n+\langle D, \vec b \rangle)\rceil$ repetitions is at most
\begin{equation*}
	\left(1-\varepsilon\right)^{\kappa \cdot r \ln(n+\langle D, \vec b \rangle)} \leq \langle D, \vec b \rangle^{-c} n^{-rc - \frac{1}{\varepsilon}} \mcomma
\end{equation*}
where the inequality follows by our choice of $\kappa$.
	Because the number of ellipsoid iterations over all iterations of the binary search is at most $\langle D, \vec b \rangle^{c} n^{rc}$, a union bound implies that with probability at most $n^{-\frac{1}{\varepsilon}}$, there is at least one moment during the whole execution of \Cref{alg:constrained_round-or-cut} when $\Pr[\ell(\vec{x})<(1 - \eps) T] \geq \varepsilon$ but no cutting plane is generated in the repeat loop.
	We finish the proof by showing that only in such a case, the algorithm may return a point $\vec{x}$ with $\E[f(P)] < \frac{T}{(d-1)(r+1)}\cdot(1-\varepsilon)^2$.
Indeed, if this case did not happen, and a vector $\vec{x}$ is returned, then 
	$\Pr[\ell(\vec{x})\geq (1 - \eps)T] \geq 1-\varepsilon$.
Thus, $\E[\ell(\vec x)]\ge (1 - \eps)^2 T$, which implies
	\begin{equation*}
		\E_{P\sim\RRR(\vec x)}[f(P)] \ge \frac{\E[\ell(\vec{x})]}{(d-1)(r+1)} \ge \frac{(1 - \eps)^2 T}{(d-1)(r+1)} \mcomma
	\end{equation*}
where the first inequality follows from \Cref{lem:round-or-cut}.

	We lose another factor of $(1-\eps)$ from terminating the binary search with a small remaining gap.
	This proves a guarantee with $(1 - \eps)^3$ instead of $1-\eps$, but we can also
	obtain the latter by rescaling $\eps$ with a constant.

By \Cref{lem:min-probability}, the optimum cannot be above the range of the binary search because none of the elements deleted at the beginning is part of any realization of $\opt$.
We recall that the binary search range is $\left[\sfrac{f(V)}{n c_D}, f(V)\right]$ and that we replaced $V$ by the set $U$ from \Cref{lem:min-probability} at the beginning of the algorithm.
If $\E_{\opt\sim \mathcal{O}}[f(\opt)]$ is below the range, then the guarantee of the solution from \Cref{lem:min-probability} suffices.

The running time of \Cref{alg:constrained_round-or-cut} follows from the fact that each single operation takes at most $n^{O(r)}$ time and the ellipsoid method terminates after at most polynomially many iterations.
\end{proof}

 \section{Submodular Markov Decision Process}\label{sec:mdp}
To simplify the transfer of our results, we deviate from the typical formalization of MDPs with states and actions as in the introduction.
The following model is equivalent but closer to Submodular Orienteering.
As before we consider a directed graph $G = (V, A)$ on layers $V = V_1\cup V_2\cup \cdots \cup V_H$.

Each vertex is either a decision vertex or a chance vertex.
We write $D$ for the decision vertices and $C$ for the chance vertices.
Thus, $V = C\cup D$.
There is one distinct chance vertex $s$ with $V_1 = \{s\}$.
There are transition probabilities $\vec p = (p_c(v))_{c\in C, v\in V}$ that describe the probability that the successor of some $c\in C$ is
$v\in V$.
We assume that $p_c(v)$ is non-zero if and only if there is an arc $(c, v)$ in the graph.
A (randomized) policy $\pi$ is a function that assigns to any \emph{trajectory} $(v_1,v_2,\dotsc,v_\ell)$, which is a walk in $G$ starting in $v_1 = s$ and ending in some decision vertex $v_\ell\in D$, a probability distribution over vertices.
We write $\pi_{v_{\le \ell}}(v)$ for the probability that $v$ is chosen as the next vertex based on the trajectory $v_{\le \ell} = (v_1,v_2,\dotsc,v_{\ell})$.
We require that $\pi_{v_{\le \ell}}(v) = 0$ unless there is an arc $(v_{\ell}, v)$ in the graph.
We are mainly interested in implicit policies that can be computed fast, more precisely in polynomial or quasi-polynomial time, on a given trajectory.

We consider here a monotone submodular reward function $f \colon 2^{V} \to \RR$.
We associate with the MDP and a policy the expectation of the reward $f(\{v_1,v_2,\dotsc,v_{\LH}\})$
over the vertices that are visited in the stochastic process where $v_1 = s$,
and $v_i$ for $i > 1$, is chosen either based on $p_{v_{i-1}}$, if $v_{i-1}$ is a chance vertex, or based on $\pi$, otherwise. 

We assume without loss of generality that $\LH = d^r$ for some $d\in\Z_{\ge 2}$ and $r\in \Z_{\ge 1}$.
The equivalence to the model in the introduction can be seen as follows. For each time and state, $s_t$,
we have a decision vertex $v \in V_{2t}$. For each time, state, and action,
$s_t$ and $a_t$, we have a chance vertex $u\in V_{2t + 1}$.
Furthermore, there is a chance vertex $s$ with $\{s\} = V_1$, which is distinct from all other vertices.
The transition probability $p_s(v)$ is the probability of state $v$ to be the initial state and
the transition probabilities $p_{u}(v)$, $u\in C\setminus\{s\}$, come from the transition probabilities
of a time, state and action.
If the number of layers is not a power of $d$, we add dummy layers. 
This transformation increases $H$ only by a constant factor, which is insignificant for our results. 

\paragraph{The LP formulation.}
To obtain an LP formulation that captures policies of MDPs, we use $Q^{(r)}$ with the following additional constraints, which intuitively enforce that on chance vertices, the successor vertex is chosen according to the given transition distribution:
\begin{equation}\label{eq:lp-mdp}
	x_{I\cup \{u\}} = x_{I} \cdot p_c(u) \qquad \begin{aligned}
		&\forall i\in\{1,2,\dotsc,\LH-1\}, c\in V_i\cap C, u\in V_{i+1} \mcomma \\
		&\{c\}\subseteq I\subseteq V_1\cup\dotsc\cup V_{i}, |I|\le r \mperiod
	\end{aligned}
\end{equation}
Here, we think of $x_I$ and $x_{I\cup\{u\}}$ as the probabilities that all vertices in $I$ (or in $I\cup\{u\}$)
are visited.

In the statement below, as well as later, when talking about ``the optimal policy'', we mean any fixed optimal policy in case there are multiple optimal policies.

\begin{lemma}\label{lem:mdp-opt}
	Let $\opt \sim \mathcal O$ where $\mathcal O$ is the distribution over layer-spanning paths arising from the MDP process with the optimal policy, that is, applying the optimal policy
	on decision vertices and the transition distributions $\vec p$ for chance vertices.
	Then $\E[\vec 1^{(r)}_{\opt}]$ satisfies~\eqref{eq:lp-mdp}.
\end{lemma}
\begin{proof}
	Let $c\in V_i\cap C$, $u\in V_{i+1}$, and $\{c\} \subseteq I\subseteq V_1\cup\cdots\cup V_{i}$, $|I|\le r$.
	If $\Pr[I\subseteq \opt] = 0$ then also $\Pr[I\cup\{u\}\subseteq \opt] = 0$.
	Thus,
	\begin{equation*}
		\E[(\vec 1^{(r)}_\opt)_{I\cup \{u\}}] = 0 = \E[(\vec 1^{(r)}_\opt)_{I}]
	\end{equation*}
	and both sides in~\eqref{eq:lp-mdp} are zero.
	Now assume that $\Pr[I\subseteq \opt] > 0$.
	Then
	\begin{equation*}
		\E[(\vec 1^{(r)}_\opt)_{I\cup \{u\}}] = \Pr[I\cup\{u\}\subseteq \opt]
		= \Pr[I\subseteq \opt] \cdot p_c(u)
		= \E[(\vec 1^{(r)}_\opt)_{I}] \cdot p_c(u) \mperiod \qedhere
	\end{equation*}
\end{proof}

\paragraph{The policy.}
We will prove that the policy in \Cref{alg:mdp-policy} (together with $p$) produces a distribution identical to
$\RRR(\vec x)$ and from this we derive the claimed approximation guarantee.
For convenience, we define the policy in \Cref{alg:mdp-policy} for any layer $\ell$ and any trajectory $(v_1,\dotsc,v_\ell)$, even if $v_\ell$ is a chance node.
In this case, the policy is not actually used, but it is still well-defined.
Note that by~\eqref{eq:lp-mdp}, \Cref{alg:mdp-policy} actually simulates the environment's probability distribution on chance nodes.
This is convenient later for the analysis.
\begin{algorithm2e}
\DontPrintSemicolon
\caption{Submodular MDP Policy}\label{alg:mdp-policy}
\KwIn{Layered MDP graph $G=(V,A)$ with $V=V_1 \cup \cdots \cup V_{\LH}$, $\LH=d^r$, $\vec x\in Q^{(r)}$ satisfy \eqref{eq:lp-mdp}, and trajectory $(v_1,v_2,\dots, v_\ell) \in V_1 \times V_2 \times \dots \times V_\ell$.}
	\KwOut{Distribution over vertices in $V_{\ell+1}$ that are neighbors of $v_\ell$.}

	Let $K = \{ \max\{k \cdot d^j \mid k\in\ZZ, k \cdot d^j \le \ell \} \mid j\in\ZZ\} \setminus\{0\}$\;
	Let $I = \{ v_i : i\in K \}$\;
	\If{$x_I = 0$}{
		\Return ``error''
	}\Else{
		\Return $\left(\frac{x_{I\cup\{u\}}}{x_I}\right)_{u\in V_{\ell+1}}$
	}
\end{algorithm2e}

We first argue that if we iteratively sample a path using only \Cref{alg:mdp-policy}, then we never encounter an error.
Because \Cref{alg:mdp-policy} defines a policy that simulates on chance nodes the environment's probability distribution, the lemma below also applies to running the policy of \Cref{alg:mdp-policy} only on decision nodes and sampling successors of chance nodes according to $\vec p$.
\begin{lemma}
	Let $\vec x\in Q^{(r)}$ satisfy~\eqref{eq:lp-mdp}.
	If \Cref{alg:mdp-policy} is called with some trajectory $v_1,\dotsc,v_{\ell}$ ending in a chance vertex $v_{\ell}\in C$ and with $\vec x$, then either ``error'' is returned or $v_{\ell+1}$ is chosen according to $p_{v_{\ell}}$.
\end{lemma}
\begin{proof}
	Assume that no error is returned and therefore $x_I > 0$ where $I$ is as in \Cref{alg:mdp-policy}.
	By~\eqref{eq:lp-mdp} we have for each $u\in V_{\ell+1}$ that
	\begin{equation*}
		x_{I\cup \{u\}} = x_I \cdot p_{v_{\ell}}(u) \mperiod
	\end{equation*}
	Thus, $u$ is chosen with probability $x_{I\cup\{u\}} / x_I = p_{v_{\ell}}(u)$.
\end{proof}
\begin{lemma}
	Let $\vec x\in Q^{(r)}$ satisfy~\eqref{eq:lp-mdp}.
	Suppose we obtain $v_1,\dotsc,v_H$ by repeatedly applying \Cref{alg:mdp-policy} on the previous vertices and $\vec x$. Then
	the policy never outputs ``error''.

\end{lemma}
\begin{proof}
	In the first iteration ($\ell=0$) we have that $K = I = \emptyset$. Thus, $x_I = x_\emptyset = 1$ and
	no error is returned.
	Suppose now that $v_1,\dotsc,v_\ell$ have been sampled without an error. We will argue
	that in the next iteration we also do not return an error.
	Let $K$ and $I$ be as in the algorithm in the iteration where $v_{\ell}$ was sampled and
	$K'$ and $I'$ as in the next iteration.
	Then $x_{I\cup \{v_{\ell}\}} > 0$, since otherwise $v_{\ell}$ would have a zero probability of
	being sampled.
	Furthermore, $K' \subseteq K\cup\{\ell\}$.
	It follows that $x_{I'} \ge x_{I\cup\{v_{\ell}\}} > 0$. Thus, also in the next iteration no error is returned.
\end{proof}

In order to use our results from the previous section, we relate the policy to our recursive randomized rounding algorithm.
\begin{lemma}\label{lem:mdp-rrr}
	Let $\vec x\in Q^{(r)}$ satisfying~\eqref{eq:lp-mdp}.

	Let $\mathtt{MDP}(\vec x)$ be the distribution over layer-spanning paths obtained from repeatedly applying \Cref{alg:mdp-policy} on the previous vertices and $\vec x$. 
	Then
	$\mathtt{MDP}(\vec x) = \RRR(\vec x)$.
\end{lemma}
\begin{proof}
	For a set of layer indices $L\subseteq\{1,\dotsc,H\}$, let $V[L]\coloneqq\bigcup_{i\in L} V_i$ and
	let $G[L]$ be the induced subgraph on $V[L]$.

	Consider the structure of the recursive calls in \Cref{alg:recursive_rounding}.
	Each recursive call works on a subgraph $G[L_{j,h}]$
	where
	\begin{equation*}
		L_{j,h} \coloneqq \{(j-1) d^h + 1, (j-1) d^h+2,\dotsc, jd^h\}
	\end{equation*}
	for some $h\le r$ and $j\in \{1,\dotsc,d^{r-h}\}$.
	The call corresponding to $L_{j,h}$ recurses on $L_{j',h-1}$ for $j'\in\{dj-d+1,dj-d+2,\dotsc,dj\}$, in increasing order of $j'$. This corresponds to a partition of $L_{j,h}$ into $d$ consecutive parts of equal length $d^{h-1}$.

	Let $v_1\in V_1,\dotsc,v_H\in V_H$ be random variables corresponding to the selected vertices
	in $\RRR(\vec x)$.
	Define
	\begin{align*}
		K_{j,h} &\coloneqq \{k_{j,h}^{(i)} \mid i\in\ZZ\} \setminus\{0\}, \text{ where } k_{j,h}^{(i)} \coloneqq \max\{k\cdot d^i \mid k\in\ZZ, k\cdot d^i\le (j-1) d^h \} \mcomma \text{\ \ and} \\
		I_{j,h} &\coloneqq \{v_i \mid i\in K_{j,h}\} \mperiod
	\end{align*}
	We prove inductively that the LP solution we pass to the subproblem on $L_{j,h}$ is
	\begin{equation*}
		\vec x^{|I_{j,h}} \coloneqq \left(\frac{x_{I_{j,h}\cup J}}{x_{I_{j,h}}}\right)_{J\in \binom{V[L_{j,h}]}{\le h+1}} \mcomma
	\end{equation*}
	Note that the solution is always projected to variables of vertices in $V[L_{j,h}]$.
	For simplicity we omit this projection in the arguments below.

	For the root problem $L_{1,r}$, note that $K_{1,r} = \emptyset = I_{1,r}$ and therefore
	$\vec x^{|I_{1,r}} = \vec x$, which is indeed the input of the problem.
	Consider now some subproblem $L_{j,h}$ and assume that $\vec x^{|I_{j,h}}$ is its input.
	Let $j'\in \{dj-d+1,\dotsc,dj\}$.
	\paragraph*{Case 1: $j' = dj-d+1$.} Then \Cref{alg:recursive_rounding}
	calls subproblem $L_{j',h-1}$ with $\vec x^{|I_{j,h}}$, that is, it does not perform any additional conditioning.
	Note that
	\begin{align*}
		k^{(i)}_{j',h-1} &= \max\{k\cdot d^i \mid k\in\ZZ, k\cdot d^i\le (j'-1) d^{h-1} \} \\
		&= \max\{k\cdot d^i \mid k\in\ZZ, k\cdot d^i\le (dj-d) d^{h-1} \} \\
		&= \max\{k\cdot d^i \mid k\in\ZZ, k\cdot d^i\le (j-1) d^{h} \}
		= k^{(i)}_{j,h} \mperiod
	\end{align*}
	Thus, $\vec x^{|I_{j',h-1}} = \vec x^{|I_{j,h}}$,
	which proves the induction step for this case.

	\paragraph*{Case 2: $j' > dj-d+1$.} Then \Cref{alg:recursive_rounding} calls 
	$L_{j',h-1}$ with $(\vec x^{|I_{j,h}})^{|u}$, where $u=v_{(j'-1)d^{h-1}}$ is the last vertex
	selected in the subproblem $L_{j'-1,h-1}$.
	Note that $(j - 1) d^h < (j'-1)d^{h-1} < j d^h$. Thus, for every $i\ge h$ we have
	\begin{align*}
		k^{(i)}_{j',h-1} &=\max\{k\cdot d^i \mid k\in\ZZ, k\cdot d^i\le (j'-1) d^{h-1} \} \\
		&= \max\{k\cdot d^i \mid k\in\ZZ, k\cdot d^i\le (j-1) d^h \} = k^{(i)}_{j,h} \mperiod
	\end{align*}
	Furthermore, for all $i<h$ we have
	\begin{align*}
		k^{(i)}_{j',h-1} = \max\{k\cdot d^i \mid k\in\ZZ, k\cdot d^i\le (j'-1) d^{h-1} \} = (j'-1)d^{h-1} \mcomma
	\end{align*}
	which is the index of the layer of $u$. On the other hand,
	\begin{align*}
		k^{(i)}_{j, h} = \max\{k\cdot d^i \mid k\in\ZZ, k\cdot d^i\le (j-1) d^h \} = (j-1) d^h = k^{(h)}_{j,h} \mperiod
	\end{align*}
	It follows that
	\begin{align*}
		K_{j',h-1} &= \{k^{(i)}_{j',h-1} \mid i\in\ZZ\} \setminus\{0\} \\
		&= \{k^{(i)}_{j,h} \mid i\in\ZZ\} \setminus\{0\} \cup \{(j'-1) d^{h-1}\} = K_{j,h}\cup \{(j'-1) d^{h-1}\} \mperiod
	\end{align*}
	Thus, $I_{j',h-1} = I_{j,h} \cup \{u\}$.
	The induction step now follows from
	\begin{equation*}
		(\vec x^{|I_{j,h}})^{|u}_J = \left(\frac{x^{|I_{j,h}}_{J\cup \{u\}}}{x^{|I_{j,h}}_u}\right)
		= \left(\frac{x_{J\cup I_{j,h} \cup \{u\}} / x_{I_{j,h}}}{x_{I_{j,h}\cup \{u\}} / x_{I_{j,h}}}\right) = \vec x^{|I_{j,h}\cup\{u\}}_J = \vec x^{|I_{j',h-1}}_J \mperiod
	\end{equation*}
	The vertex $v_{\ell+1}$ in a layer $V_{\ell+1}$ is chosen in the deepest recursion and
	$v_1\in V_1,\dotsc,v_{\LH}\in V_{\LH}$ are sampled in this order.
	By the induction above, vertex $v_{\ell+1}$ is chosen according to $\vec x^{|I_{\ell+1,0}}$
	and $K_{\ell+1,0} = \{\{\max k\cdot d^i \mid k\in\ZZ, k\cdot d^i \le \ell\} \mid i\in\ZZ\}\setminus \{0\}$.
 	This is therefore identical to \Cref{alg:mdp-policy}.
\end{proof}

\begin{proof}[Proof of \Cref{thm:main-mdp}]
	Let $\mathcal O$ be defined as in \Cref{lem:mdp-opt} and $\opt\sim \mathcal O$.
	Compute
	$\vec x\in Q^{(r)}$ satisfying~\eqref{eq:lp-mdp} using \Cref{lem:getXWithConstraints}.
	Assume that \Cref{lem:getXWithConstraints} is successful, which happens with probability at least $1 - n^{-\sfrac 1 \eps}$.
	Then 
	\begin{equation*}
		\E_{P\sim \RRR(\vec x)}[f(P)] = \E_{P\sim \mathtt{MDP}(\vec x)}[f(P)] \ge \frac{1 - \eps}{(d-1)(r+1)} \E[f(\opt)] \mcomma
	\end{equation*}
	where the equality follows from \Cref{lem:mdp-rrr} and the inequality from \Cref{lem:mdp-opt} combined
	with \Cref{lem:getXWithConstraints}.
	The number of vertices from the trajectory that a decision depends on is the size of the set we condition
	on in \Cref{alg:mdp-policy}, which is at most $r$.
	We can set $d = 2$ and $r = O(\log \LH)$, or $d = \LH^\eps$ and $r = O(1/\eps)$ to obtain the trade-offs claimed in the theorem.
\end{proof}
 \section{Conclusion}

One key open question is whether the guarantees we obtain with our quasi-polynomial time procedures can also be obtained with polynomial time algorithms.
Already for Submodular Orienteering, the question of whether the Recursive Greedy algorithm of \textcite{chekuri2005recursive} can be improved to a polynomial time algorithm with similar guarantees is a long-standing open problem.
Another natural direction is to get a better understanding of what trade-offs are possible between the size of a policy and its performance for Submodular MDPs, also from an information-theoretic hardness perspective.

Moreover, we want to highlight that one can further thin out variables of our LP relaxation $Q^{(r)}$, which we deliberately did not do in order to keep the paper concise and focused.
More precisely, the Sherali-Adams type LP relaxation $Q^{(r)}$ that we introduced has a variable for every subset of vertices of size at most $r+1$.
It suffices to only consider variables for subsets of vertices that are relevant for later conditioning steps.
One can observe that this corresponds to only including variables for sets $K$ as defined in \Cref{alg:mdp-policy} together with one extra variable, and subsets thereof.
This change makes the notation more cumbersome, and our quasi-polynomial time procedure would still be quasi-polynomial time.
One advantage of this change, apart from leading to a smaller LP, is that our rounding would be marginal-preserving for all variables in the LP, whereas with our current LP, we only used the marginal-preserving property for single variables (see \Cref{lem:marginal-preserving}~\ref{en:marginals}), which was enough for our purposes.
 
\printbibliography

\appendix
\crefalias{section}{appendix}
\crefalias{subsection}{subappendix}

\section{Combinatorial \texorpdfstring{$O(n^\eps)$}{O(n^epsilon)}-approximation for Submodular Orienteering}
\label{sec:comb_neps}
Recall that the Recursive Greedy algorithm yields a logarithmic approximation in quasi-polynomial time.
Our LP-based method also gives an $O(n^\eps)$-approximation in time $n^{O(1/\eps)}$.
Here, we show that the latter can also be achieved with a simpler combinatorial algorithm.
However, this algorithm does not yield the flexibility we get with the LP-based method
and therefore, to the best of our knowledge, it cannot be used to obtain the MDP application presented
in this paper. To ease the presentation we omit the length constraint in the algorithm below. We comment on
how one could integrate it at the end of the section.

The combinatorial algorithm is based on unpublished work of the first author together with Nick Fischer, who kindly
gave us permission to include it.

As in the other algorithm, we assume that we are given a layered instance
with $H = d^r$ layers $V_1\cup\cdots\cup V_H$. For the reduction to this case, we refer to
the proof of \Cref{thm:main-orienteering}.
To simplify recursive calls, we additionally have inputs $S\subseteq V_1$ and $T\subseteq V_H$,
and the goal is to find a path from some vertex in $S$ to some vertex in $T$, maximizing $f$. 

The algorithm uses a combination of recursion and dynamic programming and is given in \Cref{alg:recursive_dp}.
For $d=2$, it behaves as Recursive Greedy from~\cite{chekuri2005recursive}.
We use the following notation in the algorithm. For two paths $P, P'$
such that the last vertex of $P$ has an arc to the first vertex of $P'$, we denote by $P\circ P'$ the
concatenation of $P$ and $P'$.
We write $N^+(v)$ as the out-neighborhood of $v$, that is, all $u\in V$ such that $(v,u)\in A$.
For some $L\subseteq \{1,2,\dotsc,H\}$, we write $G[L]$ as the induced subgraph of $G$ on vertices
$\bigcup_{i\in L} V_i$.

\begin{algorithm2e}
\DontPrintSemicolon
	\caption{Recursive dynamic program}\label{alg:recursive_dp}
	\KwIn{Layered graph $G=(V,A)$ with $V=V_1 \cup \cdots \cup V_{\LH}$, $S\subseteq V_1$, $T \subseteq V_H$, $d^r = \LH$.}
\KwOut{$S$-$T$ path (or $\bot$ if no such path exists).}

	\If{$r = 0$}{
		\Return $\mathrm{argmax} \{f(P) \mid P = (v), v\in S\cap T\}$ or $\bot$ if $S\cap T = \emptyset$\;
	}

	Let $L_i = \left\{\frac{(i-1)H}{d}+1,\dotsc,\frac{iH}{d}\right\}$ for all $i\in\{1,2,\dotsc,d\}$\;
	\For{$v\in V_{H/d}$}{
		Recursively compute $S$-$v$ path $P_v$ in $G[L_1]$ with $f$\;
		$D[1, v] = P_v$\;
	}
	\For{$i = 2,3,\dotsc,d$}{
		Let $D[i,v] = \bot$ for all $v\in V_{i H/d}$ \;
		\For{$u\in V_{(i-1) H/d}$ and $v\in V_{i H/d}$ with $D[i-1,u]\neq \bot$}{
			Recursively compute $N^+(u)$-$v$ path $P_{u,v}$ in $G[L_i]$ with function $f(\;\cdot \mid D[i-1,u])$ \;
			\If{$P_{u,v} \neq \bot$}{
				\If{$D[i,v] = \bot$ or $f(D[i-1,u] \circ P_{u,v}) > f(D[i,v])$}{
					$D[i,v] = D[i-1,u] \circ P_{u,v}$\;
				}
			}
		}
	}
	\Return $\mathrm{argmax} \{f(P) \mid P = D[d,v], P\neq \bot, v\in T\}$ or $\bot$ if the set is empty \;
\end{algorithm2e}

\begin{lemma}
	Let $G = (V,A)$ be a layered graph with $V = V_1\cup\cdots\cup V_H$, $S\subseteq V_1$, $T\subseteq V_H$, and $d^r = H$. Assume that there is a path from $S$ to $T$ and let $\opt$ be the path maximizing $f(\opt)$.
	Let $P$ be the path returned by \Cref{alg:recursive_dp} on these parameters. Then $P\neq \bot$ and
	\begin{equation*}
		f(P) \ge \frac{1}{(d-1)(r+1)} f(\opt) \ .
	\end{equation*}
	Furthermore, the running time is $n^{O(r)}$.
\end{lemma}
\begin{proof}
	Note that the degree of the recursion tree
	is at most $(d-1) n^2$. 
	The depth of the recursion tree is $r$. Therefore, the number of nodes of the recursion tree
	is $(dn)^{O(r)}\le n^{O(r)}$. The number of operations in the algorithm (without the recursions)
	is polynomial. This proves the claimed running time.

	For the approximation guarantee, we argue via induction over $r$.
	For $r = 0$, the algorithm is optimal
	and since $(d-1)(r + 1) \ge 1$ the claim follows.
	Now assume that $r \ge 1$.
	For $i\in\{1,\dotsc,d\}$, let $\opt[L_i]$ denote the restriction of $\opt$ to $G[L_i]$, with $L_i$ as
	in the algorithm, and let $v_i$ be the vertex in the intersection of $\opt$ and $V_{iH/d}$.
	By the induction hypothesis, we have for all $i\in\{1,\dotsc,d\}$ that $D[i,v_i]\neq \bot$ and
	\begin{equation*}
		f(D[i,v_i]) \ge \begin{cases}
			\frac{f(\opt[L_1])}{(d-1)r} &\text{ if } i = 1, \\
			f(D[i-1,v_{i-1}]) + \frac{f(\opt[L_i] \mid D[i-1,v_{i-1}])}{(d-1)r} &\text{ if } i \ge 2.
		\end{cases}
	\end{equation*}
	Note that, in particular, $f(D[i,v_i])$ is non-decreasing in $i$.
	By adding the inequalities over all $i\in\{1,\dotsc,d\}$, we obtain
	\begin{align*}
		f(D[d,v_d]) &\ge \frac{1}{(d-1)r} \left(f(\opt[L_1]) + \sum_{i=2}^d f(\opt[L_i] \mid D[i-1,v_{i-1}])\right) \\
		&\ge \frac{1}{(d-1)r} \left(f(\opt[L_1]) + \sum_{i=2}^d \left[f(\opt[L_i]) - f(D[i-1,v_{i-1}])\right]\right) \\
		&\ge \frac{1}{(d-1)r} \left(f(\opt) - (d-1) f(D[d,v_{d}])\right) \mperiod
	\end{align*}
	By moving the term in $f(D[d,v_d])$ to the left and multiplying with $r/(r+1)$ we conclude
	\begin{equation*}
		f(P) \ge f(D[d,v_d]) \ge \frac{1}{(d-1)(r+1)} f(\opt) \mperiod \qedhere
	\end{equation*}
\end{proof}
By setting $d = n^\eps$ and rescaling $\eps$, we obtain the following.
\begin{theorem}
	For $\epsilon >0$, there is a combinatorial $O(n^\eps)$-approximation algorithm for Submodular Orienteering (without a length bound) with running time $n^{O(1/\eps)}$.
\end{theorem}
One could integrate a length function
in a similar way as in~\cite{chekuri2005recursive}, resulting in essentially the same result also 
including a length bound. This works by rounding the objective values and storing
for each rounded objective value the solution of lowest length.
We omit the details for the sake of simplicity.

 \end{document}